\documentclass[11pt,reqno]{article}

\usepackage{amsmath,amssymb,amsthm,amsfonts} 
\usepackage{bm}
\usepackage{enumerate}
\usepackage{mathtools}
\usepackage{nicefrac}

\newcommand{\ifndef}[2]{\ifthenelse{\isundefined{#1}}{#2}{}}

\newtheoremstyle{mydef}{10pt}{1pt}{\slshape}{0pt}{\bfseries}{.}{3pt}{}
\theoremstyle{mydef}
\newtheorem{theorem}{Theorem}[section]

\newtheorem{lemma}[theorem]{Lemma}
\newtheorem{proposition}[theorem]{Proposition}
\newtheorem{cor}[theorem]{Corollary}
\theoremstyle{definition}
\newtheorem{definition}[theorem]{Definition}
\newtheorem{question}[theorem]{Question}
\newtheorem{example}[theorem]{Example}
\newtheorem{notation}[theorem]{Notation}

\newtheorem{remark}[theorem]{Remark}
\newtheorem{fact}[theorem]{Fact}

\ifndef{\claim}{\newtheorem{claim}[theorem]{Claim}}

\usepackage{hyperref}
\hypersetup{colorlinks=true,
linkcolor=black,          
citecolor=black,
}

\DeclarePairedDelimiter{\bracket}{[}{]}
\DeclarePairedDelimiter{\set}{\{}{\}}
\DeclarePairedDelimiter{\abs}{\lvert}{\rvert}
\DeclarePairedDelimiter{\norm}{\lVert}{\rVert}
\DeclarePairedDelimiter{\floor}{\lfloor}{\rfloor}
\DeclarePairedDelimiter{\ceil}{\lceil}{\rceil}

\newcommand{\NN}{\mathbb{N}}
\newcommand{\RR}{\mathbb{R}}
\newcommand{\CC}{\mathbb{C}}

\newcommand{\e}{\emph}

\DeclareMathOperator*{\Ex}{ \bf E}
\date{}

\newcommand{\Eop}{\mathop{\mathbf{E}}\nolimits}
\makeatletter
\newcommand{\E}{\@ifstar\E@star\E@nostar}
\newcommand{\E@nostar}{\@ifnextchar[\E@size\E@plain}
\newcommand{\E@star}{\def\@Esize{*}\@ifnextchar_\E@sub\E@nosub}
\def\E@size[#1]{\def\@Esize{[#1]}\@ifnextchar_\E@sub\E@nosub}
\newcommand{\E@plain}{\def\@Esize{}\@ifnextchar_\E@sub\E@nosub}
\def\E@sub_#1{\Eop_{#1}\expandafter\bracket\@Esize}
\newcommand{\E@nosub}{\Eop\expandafter\bracket\@Esize}
\makeatother

\let\newmat\newcommand
\let\dt\cdot
\newcommand{\dc}{,\dots,}
\def\01{\set{0,1}}
\newcommand{\U}{\mathcal{U}}

\newcommand{\RII}{\mathit{SMP}}
\DeclareMathOperator{\PR}{Pr}
\newcommand{\deq}{\stackrel{\mathrm{def}}{=}}
\let\tm\cdot

\newcommand{\asO}[1]{O(#1)}
\newcommand{\asOm}[1]{\Omega(#1)}

\newcommand{\Cl}[1]{\mathcal{#1}}

\newcommand{\F}{\mathbb{F}}

\newcommand{\IP}{\mathit{IP}}

\newmat{\poly}{\mathrm{poly}}

\newcommand{\GIP}{\mathit{GAPIP}}
\newcommand{\RAIP}{\mathit{RIP}}

\newmat{\RA}{\mathit{RA}}
\newmat{\RAd}{\mathit{RA}_{n,\delta}}
\newmat{\Eq}{\mathit{EQ}}

\newmat{\Ra}{R_A}
\newmat{\Rb}{R_B}
\newmat{\bRa}{\bar R_A}
\newmat{\bRb}{\bar R_B}
\newmat{\Ma}{M_A}
\newmat{\Mb}{M_B}
\newmat{\ra}{r_A}
\newmat{\rb}{r_B}
\newmat{\ma}{m_A}
\newmat{\mb}{m_B}
\newmat{\ro}{\rho_\vee}

\DeclareMathOperator{\col}{col}
\DeclareMathOperator{\agr}{agr}
\DeclareMathOperator{\Cor}{Cor}
\DeclareMathOperator{\Var}{Var}

\usepackage{titlesec}

\usepackage[left=1in, right=1in, top=1in, bottom=1.2in, headsep=0.5cm]{geometry}
\titlespacing*{\section}{0pt}{7pt}{8pt}
\titlespacing*{\subsection}{0pt}{3pt}{3pt}
\titlespacing*{\paragraph}{0pt}{5pt}{3pt}

\titleformat{\section}[block]
{\Large \bfseries}
{\thesection}{12pt}{}{}

\titleformat{\subsection}[block]
{\large \bfseries}
{\thesubsection.}{3pt}{}[]

\titleformat{\paragraph}[runin]
{\large \bfseries}
{}{}{}[]

\begin{document}

\title{  \bfseries \LARGE On the Role of Shared Randomness in Simultaneous Communication\footnote{A preliminary version of this article appeared in the proceedings of ICALP 2014, pages  150-162.}}

\newcommand{\instMB}{Massachusetts Institute of Technology, Cambridge, MA, U.S.A. Partially supported by NSF STC Award 0939370.}

\newcommand{\instDG}{Institute of Mathematics, Czech Academy of Sciences, Praha, Czech Republic.}

\newcommand{\instNEC}{NEC Laboratories America, Princeton, NJ, U.S.A.}

\newcommand{\thanksDG}{Partially funded by the grant P202/12/G061 of GA \v CR and by RVO:\ 67985840.}

\newcommand{\thaNEC}{Part of this work was done while at the \instNEC}

\author{
 Mohammad Bavarian\thanks{\instMB\ \thaNEC}
 \and  Dmitry Gavinsky \thanks{\instDG\ \thanksDG\ \thaNEC}  
 \and  Tsuyoshi Ito \thanks{\thaNEC}
}

\maketitle

\begin{abstract}
Two parties wish to carry out certain distributed computational tasks, and they are given access to a \e{source of correlated random bits}.
It allows the parties to act in a correlated manner, which can be quite useful.
But what happens if the shared randomness is not perfect?

In this work, we initiate the study of the power of different sources of shared randomness in communication complexity.
This is done in the setting of \e{simultaneous message passing (SMP) model of communication complexity}, which is one of the most suitable models for studying the resource of shared randomness. 
Toward characterising the power of various sources of shared randomness, we introduce a measure for the ``quality'' of a source -- we call it \emph{collision complexity}.
Our results show that the collision complexity tightly characterises the power of a (shared) randomness resource in the SMP model.

Of independent interest is our demonstration that even the ``weakest'' sources of shared randomness can in some cases increase the power of SMP substantially:\ the \emph{equality} function can be solved very efficiently with virtually any nontrivial shared randomness.

\end{abstract}

\section{Introduction}

One of the central themes of complexity theory is the study various resources of computation and their respective power.
In this direction many different models of computation have been defined, and various resources have been considered -- e.g., time, space, randomness, access to powerful provers, etc.
Among the strongest models where the researchers have been able to prove strong lower bounds
is the model of \e{communication complexity}.

An important resource in communication complexity and, more generally, in virtually all distributed computational settings, is the access to a source of \emph{correlated random bits} or \e{shared randomness (SR)}:
It allows the parties to act in a correlated way in order to perform better and more efficiently in their tasks.

The primary goal of this work is a quantitative characterisation of the utility of different sources of shared randomness to parties interested in solving communication problems. To this end, let us define sources of shared randomness more formally.

\begin{definition}[Source of SR for two parties]\label{def:SR_source}
A source of shared randomness (also referred to as a bipartite distribution) is specified by a distribution $\rho$ on a domain $U\times V$, where  $U$ and $V$ correspond to the parts of shared randomness visible to Alice and to Bob, respectively.
To solve their communication task, the parties are allowed to take as many independent samples as they need from the source $\rho$. 
\end{definition}


To understand the motivation behind the above definition better, it would be useful to consider a few examples of sources of shared randomness. Let  $U=V=\{0,1\}$.
The following two sources on $U\times V$ are the well-known settings of \emph{private} and \emph{perfect} forms of shared randomness, respectively:
\[ \rho_{\it priv}:=\mathcal U_{\{00,01,10,11\}},\qquad  \rho_{\it perf}:=\mathcal U_{\{00,11\}}, \]
where $\mathcal U_{\{\ldots\}}$ is defined by:
\begin{notation}
For a finite set $Q$, we denote the uniform distribution over it by $\mathcal U_Q$ .
\end{notation}
The above two examples ($\rho_{\it perf}$ and $\rho_{\it priv}$) are valuable as they indicate that the more familiar settings of communication complexity are given as special cases of our model by setting $\rho=\rho_{\it priv}$ or $\rho=\rho_{\it perf}$.
However, more interesting questions in this work are related to the intermediate cases between the ``extremes'' of $\rho_{\it perf}$ and $\rho_{\it priv}$.


\subsection{Shared randomness in SMP communication:\ definitions and more examples}

Communication complexity studies how much communication is needed to compute a given function
on distributed inputs.
The model we work with is the two-party \emph{simultaneous message passing (SMP)} model, in which the \e{players}, \e{Alice} and \e{Bob}, each send a message to the \e{referee}~\cite{NS96_Pu,BabKim97CCC}.
The SMP model augmented with a shared randomness source  $\rho$ is defined as follows: Let $\rho$ a source of shared randomness on~$U\times V$ and $(u_1,v_1), (u_2,v_2), \ldots, (u_m,v_m)$ be independent samples from $\rho$. In the SMP model with shared distribution~$\rho$, Alice receives input~$x$ and her part of shared randomness~$R_a=(u_1,\dots,u_\ell)\in U^\ell$, and Bob receives $y$ and his part of shared randomness $R_b=(v_1,\dots,v_\ell)\in V^\ell$. Alice and Bob use their parts of input and shared randomness to compute their messages, $M_a(x, R_a)$ and $M_b(y, R_b)$, to the referee who upon receiving these messages outputs an answer.
A \emph{communication protocol} determines the value of $m$ and the actions of all the participants in the protocol.

A communication protocol is said to solve a communication problem with error probability~$\delta$ if it guarantees correct answer with probability at least $1-\delta$ for every allowed input.
The communication cost of a protocol is the maximum possible total number of bits sent by the players to the referee. Given a communication problem $f$, we denote the minimum cost of a protocol that solves~$f$ (with error probability $<\frac{1}{3}$) in the SMP model with SR source $\rho$ by $\textit{SMP}_\rho(f)$.


\paragraph{Our goal.} In this work we consider the case when $\rho$ is not a perfect source of shared randomness and investigate to what extent Alice and Bob can still use $\rho$ to reduce the communication cost (as compared to the case of private randomness). Moreover, we are interested in understanding \e{what properties of the source $\rho$ determine its utility as a resource in the SMP model}.
Consider the following slightly more complicated distributions.
\begin{example}\label{example:disj} Let $U=V=\{0,1\}$.
\begin{itemize} 
\item \textbf{``Shared disjointness''}: The shared distribution is $\rho_{\it disj}:=\mathcal U_{\{00,01,10\}}$.
\item \textbf{Symmetric Noisy Bits}: The shared distribution is $\rho_{\it BSC_p}= (1-2p) \rho_{\it perf}+ 2p \rho_{\it priv}$, i.e. $\rho_{\it BSC_p}(0,0)= \rho_{\it BSC_p} (1,1)=\frac{1-p}{2}$  and  $\rho_{\it BSC_p}(0,1)= \rho_{\it BSC_p}(1,0)=\frac{p}{2}$.
\end{itemize}
\end{example}

The subtlety of the problem of determining the power of a general source of shared randomness $\rho$ becomes more clear when one notices that even in the simple case of $\rho_{\it disj}$ vs.~$\rho_{\it BSC_p}$ it is a priori unclear which source is stronger than the other (say, for $p=0.8$). Even more, it is not obvious that it should possible to make a useful comparison between two different sources of shared randomness $\rho_1$ and $\rho_2$ which, like the above case of $\rho_{\it disj}$ vs.~$\rho_{\it BSC_p}$, may have quite different forms.\footnote
{Compare this to the case $\rho_{\it BSC_p}$ vs.~$\rho_{\it BSC_q}$ where evidently when $|p|<|q|$ the first source is  more powerful than the second as the players can always use private randomness to locally simulate $\rho_{BSC_p}$ using $\rho_{BSC_{q}}$.}
Maybe surprisingly, it turns out that comparison between sources of shared randomness of very different form is possible (see Theorem \ref{thm:main}).
To achieve this, we introduce the notion of \emph{collision complexity} of a source (see Section \ref{section:collision-complexity}), which associates to a source $\rho$ a real valued function $\col_\rho:\NN \rightarrow \RR^+$, whose rate of growth captures the utility of $\rho$ as a source of shared randomness. 

Before we continue, let us exclude the following case:
\begin{definition}\label{def:degenerate_source}
	A source of shared randomness is \emph{degenerate} if it is deterministic from the point of view of at least one of the players.  
\end{definition}
In other words, a source of shared randomness is non-degenerate if the players can use it to \e{locally emulate private randomness ($\rho_{priv}$)}.
Degenerate sources are not quite that interesting for the purposes of this work as they are effectively no more powerful than the deterministic ones, and therefore throughout the paper all sources of shared randomness are assumed to be \emph{non-degenerate}.

\subsection{Our results}\label{subsection:results}
The first result is an example of an SMP protocol showing that even the weakest sources of shared randomness can be sometimes quite powerful. More specifically, we have the following theorem for the equality function~$\Eq_n\colon\{0,1\}^n\times\{0,1\}^n\to\{0,1\}$, defined to be~$1$ if~$x=y$, and~$0$ otherwise.
\begin{theorem}
  \label{theorem:eq}
  Let~$\rho$ be a probability distribution on~$U\times V$.
  If~$\rho$ is not a product distribution,
  then it holds that~$\RII_\rho(\Eq_n)= O(1)$,
  where the implied constant in $O(\cdot)$ depends only on~$\rho$ (and not on~$n$).
\end{theorem}

The above result may look somewhat surprising, as we have $\textit{SMP}_{priv}(\Eq_n)=\Theta(\sqrt{n})$ \cite{BabKim97CCC,NS96_Pu}, which means that even the presence of slightest correlation in the shared distribution $\rho$ increases the power of SMP significantly in the asymptotic sense, decreasing the communication requirements from $\Omega(\sqrt{n})$ to $O(1)$.

The above theorem gives rise the following question:
\begin{question}\label{question:odd}
Is it possible that what happened for $\Eq_n$ is generic? More precisely, is the effect of having an imperfect source of shared randomness $\rho$ instead of $\rho_{\it perf}$ limited to at most a multiplicative constant blow-up in the communication, i.e. 
\[ \textit{SMP}_\rho(f)=O_{\rho}(\hspace{1pt}\textit{SMP}_{\it perf} (f)\hspace{1pt})\]
 for any fixed non-product $\rho$?
\end{question}
Note that since  the constant in $O(\cdot)$ in the above may depend on $\rho$ (but not on the instance size $n$), the above question is  not as naive as it may at first appear.
For instance, if one trusts the (incorrect) intuition that ``usefulness'' of a source of shared randomness $(A,B)\sim\rho$ is determined by the mutual information between the two sides, $I(A;B)$, then it may appear reasonable to believe in the affirmative answer to Question~\ref{question:odd}.

Nevertheless, the answer to Question \ref{question:odd} is \emph{negative}, and the actual trade-off between $\mathit{SMP}_\rho(f)$ and $\mathit{SMP}_{\it perf}(f)$ is more subtle, and, as shown in Theorem \ref{thm:main}, is captured by the collision complexity of $\rho$. Hence, besides the technical significance of Theorem \ref{thm:main}, we believe that the fact that collision complexity (as opposed to entropic quantities like mutual information, etc.) is \e{the appropriate measure for quantifying the power of a shared random source} is perhaps a noteworthy conceptual message of the result. 

Before stating our main theorem, let us give an (informal) definition of the collision complexity and observe some of its basic properties.

\newtheorem*{undefi}{Definition}\label{def:col_intro}

\begin{undefi}[Collision complexity, informal]
For $n\in\NN$, we say that a two-player protocol \e{produces a collision} if each player output a subsets of $[n]$, respectively denoted by $A$ and $B$, such that $\Pr[i\in A\cap B]\ge1/n$ for each $i\in[n]$.\footnote{An alternate point of view on this criteria is that the parties want to ensure $\Ex[ |A\cap B|]\geq 1$, and furthermore, they want the intersecting sets $A\cap B$ to be distributed more or less evenly across $[n]$. Of course, without the latter condition achieving $\Ex[ |A\cap B|]\geq 1$ would have been easy, since the players could otherwise have always outputted a fixed $i\in [n]$; however,  the evenness condition along with the desire to minimise $|A|+|B|$ makes this a non-trivial task.} The complexity of a protocol is the maximum possible $\frac{|A|+|B|}{2}$.
The collision complexity of a source $\rho$, denoted by $\col_\rho(n)$, is the minimum complexity of a protocol that uses (arbitrary number of) samples from $\rho$ and produces a collision.
\end{undefi}
Conceptually, collision complexity quantifies the cost of simulating perfect shared randomness by two parties who only have $\rho$, and therefore it can be viewed as \e{a measure of quality} of $\rho$. For example, it is clear from the above definition that the players sharing $\rho_{\it perf}$ can always use their perfect source of shared randomness to agree on an element $i\in [n]$ and both output $A=B=\{i\}$ which shows that $\rho_{\it perf}$ has collision complexity $1$, which is optimal. The following basic properties of collision complexity also follow easily from the definition.

\begin{fact}[Properties of collision complexity] \label{fact:col} The collision complexity $\col_\rho:\NN\rightarrow \RR^+$
satisfies the following:

\begin{itemize}
\item[i.] \textbf{Basic General Bounds}: $1 \leq \col_\rho(n)\leq \sqrt{n}$. 
\item[ii.] \textbf{Private and Public Randomness}: $\col_{\rho_{\it perf}}=1$ and   $\col_{\rho_{priv}}(n)=\sqrt{n}$.
\item[iii.] \textbf{Tensor Product Stability}: $\col_{\rho^{\otimes t}}(n)=\col_\rho(n)$.
\item[iv.] \textbf{Decrease of Collision Complexity after Tensoring}: $\col_{\rho\otimes \sigma} (n) \leq \min \{\col_\rho(n), \col_\sigma(n)\}$. 
\end{itemize}
\end{fact}

Combining the property (i) and (ii) from the above, we see that the minimal value of collision complexity is achieved by  $\rho_{\it perf}$ and its maximum value is achieved by $\rho_{\it priv}$. The \textbf{\itshape main intuition} here is that the higher the ``quality'' of the correlation provided by $\rho$ is, the slower we expect the growth of $\col_\rho(n)$ to be.

We can now state our main theorem which uses collision complexity to compare $\mathit{SMP}_\rho$ with  $\mathit{SMP}_{\it perf}$ (where $\mathit{SMP}_{\it perf}$ is a shorthand for $\mathit{SMP}_{\rho_{\it perf}}$).

\begin{theorem}\label{thm:main} For any communication problem $f$ and non-degenerate source  of shared randomness $\rho$, we have 
\begin{equation}\label{eqn:main_1}
 \mathit{SMP}_\rho(f)= \widetilde{O}(\mathit{SMP}_{\it perf}(f)\, \cdot\, \col_\rho(n)\,).
\end{equation}
Moreover, there exists a family of partial function $\{g_n\}$ such that 
\begin{equation}\label{eqn:main_2}
 \mathit{SMP}_\rho(g_n)= \widetilde{\Theta}(\col_\rho(n)\,).
 \end{equation}
 \end{theorem}

The first part of Theorem \ref{thm:main} states that any source of correlation can ``replace'' the perfect source $\rho_{\it perf}=\mathcal{U}\{00,11\}$ at the multiplicative cost of $\col_\rho(n)$. The second part of the theorem guarantees the existence of a family of partial functions with $\RII_\rho(g_n)$  exhibiting the behaviour expected from the first part of the theorem -- which proves tightness of the result, up to poly-logarithmic factors.  

The  particular family of partial functions in the second part of Theorem~\ref{thm:main} is defined as follows. 

\begin{definition}[Gap-inner-product]
  \label{d_gip}
  Let $n,m\in\NN$ and $\forall i\in[n]:x_i,y_i\in \mathbb{F}_2^m$.
  Define the \emph{gap-inner-product function~$\GIP_{n,m}$} as
  \[
    \GIP_{n,m}((x_1,\dots,x_n),(y_1,\dots,y_n)) = \begin{cases}
      0, & \abs{\set{i\in[n]\colon x_i\dt y_i=0}}\ge2n/3, \\
      1, & \abs{\set{i\in[n]\colon x_i\dt y_i=1}}\ge2n/3, \\
      \bot, & \text{otherwise.}
    \end{cases}
  \]
We set~$g_n=\GIP_n$ where $\GIP_n$ is shorthand for $\GIP_{n,8\log n}$,
assuming that~$n$ is a power of two in this context.
\end{definition}

Note that the second part of Theorem \ref{thm:main} itself consists of both an upper bound and a lower bounds on $\RII_\rho(\GIP_n)$ (the main technical challenge being proving the lower bound).

Our last set of result (presented in the appendix) describes the relation between the rate of growth of $\col_\rho(n)$ and the hypercontractivity properties of $\rho$, using that to give bounds on the collision complexity of $\rho_{\it disj}$.




\section{Background and related work}

\paragraph{Simultaneous message passing model.} 
As we alluded to before, the SMP model is particularly suitable for studying shared randomness.
This is due to the fact that the power of this model depends crucially on availability of a common random source.
It is easy to see that shared randomness allows the players to use \e{mixed strategies}.

The classical example of \e{the equality function} on $n$-bit strings demonstrates a problem that requires $\Theta(\sqrt{n})$ bits of communication in the SMP model with private randomness only, but can be solved by a $O(1)$-bit protocol in the SMP model with (perfect) shared randomness.
Moreover, we will see in Section~\ref{sec:equality} that \e{any non-trivial form of shared randomness} is sufficient for solving the equality function by a $O(1)$-bit protocol.
On the other hand, in virtually all other existing two-party models, at most additive logarithmic ``saving'' in terms of communication complexity can result from the availability of shared randomness (as opposed to private randomness only).

In this work we generally view poly-logarithmic multiplicative factors as insignificant in the context of communication complexity, and therefore for our needs SMP model is the most suitable model to consider.
\footnote{
On the other hand, in a parallel work, Canonne et al.~\cite{CGMS14} study the effects of shared randomness in the one- and two-way settings in the ``sub-logarithmic regime'' and obtain rather interesting and unexpected results.
}

\paragraph{Collision complexity and measures of quality of correlation.}
Intuitively and informally, our notion of collision complexity is \e{a measure of quality of shared randomness $\rho$}as it  satisfies the tensor product stability property (Fact \ref{fact:col}). A number of other ``measures of quality'' of correlations have been extensively studied. Perhaps the most well-known among those is the \e{maximum correlation} and \e{hypercontractivity}. The latter which is of particular interest due to both its utility in a multitude of applications and its striking mathematical  elegance. 

 The literature on measures of quality of shared randomness (and in particular hypercontractivity) is vast and we will not try to give a comprehensive survey here; instead, we refer the reader to recent works  \cite{anar2013, delgosha2014, kamath2012} and references therein. We just briefly note one particular interesting line of work on a related problem called \e{non-interactive correlation distillation (NICD)} \cite{MosODo05RSA,MosODoRegSteSud06IJM,Yang07TCS}. In the two-party NICD, Alice and Bob have access to an unlimited number
of independent copies of~$\rho$ (just as in our case),
and their task is to produce a marginally uniform bit each
so that the outputs by Alice and Bob agree with the maximum probability.
Although both NICD and the collision complexity are closely related to the maximum correlation and hypercontractivity, their exact relationship is unclear. On the other hand, the relationship between hypercontractivity and maximum correlation is evident in many of the works cited above. We discuss the relationship between these two notions and collision complexity in Subsection \ref{sec:removal} and the appendix.

The direct precursor to this work is a work of Gavinsky, Ito and Wang~\cite{GavItoWan13CCC}, which, to our knowledge, was the first that studied different ``forms of shared randomness'' in communication complexity.


\section{Preliminaries}

Throughout the paper, the base of logarithm is two unless stated otherwise. For $x,y\in\01^n$, by $x\cdot y$ we mean the inner product of $x,y$ as elements of $\mathbb{F}_2^n$, i.e.~$\sum_{i=1}^n x_iy_i\bmod 2$.

Recall from the introduction that by a source of shared randomness $\rho$ (also sometimes referred to  as a \underline{bipartite distributions}) we mean a distribution over a set $\Sigma=U\times V$ and for $X\subseteq \Sigma$ the uniform distribution over a set~$X$ is denoted by ~$\U_X$. Also for any distribution $\rho$ on $U\times V$ we denote the marginals of $\rho$ on $U$ and $V$  by $\rho_{U}$ and $\rho_V$, respectively.

The main operation on sources of shared randomness  is the tensor product. 

\begin{definition}
Let $\rho_1$ and $\rho_2$ be two distributions over $U_1\times V_1$ and $U_2\times V_2$ respectively. We define a new source of shared randomness $\rho_1\otimes \rho_2$ over $(U_1\times U_2) \times (V_1\times V_2)$ by setting
\[ (\rho_1\otimes \rho_2)\, ( (x,x'),(y,y'))= \rho_1(x,y)\cdot \rho_2(x',y')\]
for any $(x,y)\in U_1\times V_1$ and $(x',y')\in  U_2  \times V_2$.
\end{definition}

\paragraph{Maximum Correlation.} Let $\rho$ be a distribution on $U\times V$. The spaces $U$ and $V$ equipped with the measures $\rho_U$ and $\rho_V$ form two probability spaces which turns the vector space of real-valued functions over $U$ and $V$ into a Hilbert space via 
$\|f\|_2^2=\Ex_{u\sim \rho_A} f(u)^2$ and   $\|g\|_2^2=\Ex_{v\sim \rho_V} g(v)^2$.

Given the above, we can define the maximum correlation of $\rho$ as follows. 

\begin{definition}[Maximum correlation]\label{def:max_cor}
Given a distribution $\rho$ on $U\times V$ the maximum correlation of $\rho$ is defined 
  \[ \Cor(\rho) = \sup_{f,g} \Ex_{(u,v)\sim\rho}[f(u)g(v)] \]
where the supremum is taken over functions satisfying $\Ex_{u\sim \rho_U} f(u)= \Ex_{v\sim \rho_V} f(v)=0$ and $\|f\|_2=\|g\|_2=1$. 
\end{definition}
From the above definition it is clear that $\Cor(\rho_{\it perf})=1$, which is clearly the largest the maximum correlation can be, because one can take $f=g$ for any $\|f\|_2=1$. On the other hand for any product distribution, such as $\rho_{\it priv}$, the maximum correlation is zero because in that case by the independence we have $\Ex_{(u,v)} f(u) f(v)= \Ex_u f(u) \Ex_{v} f(v)=0$.

The following follows easily from the definition.

\begin{lemma} \label{lemma:correlation-general}
  Let~$\rho$ be a distribution on~$U\times V$ and ~$f\colon U\to[0,1]$ and~$g\colon V\to[0,1]$. 
  Then,
  \[
    \Ex_{(u,v)\sim\rho}f(u)g(v)
    \leq 
    \Ex_u f(u) \Ex_v f(v)+\Cor(\rho)\cdot \sqrt{\Var(f)\Var(g)}. \]
\end{lemma}

\begin{proof}
Let $a=\Ex_{u\sim \rho_U} f(u)$ and $b= \Ex_{v\sim \rho_V} g(v)$. Note that $\Ex_{(u,v)\sim\rho}f(u)g(v)$ can be written as $ab+ \Ex_{(u,v)\sim\rho}\, (f(u)-a)\cdot (g(v)-b)$.
The lemma follows from noting that $(f-a)/\sqrt{\Var(f)}$
  and~$(g-b)/\sqrt{\Var(g)}$ have mean~$0$ and variance~$1$ and the definition of the maximum correlation.
\end{proof}
 
The main attractive feature of maximum correlation is its simple behaviour under \emph{tensor product} of different distributions.\footnote{Collision complexity itself also satisfies a nice property under tensor product (though this is essentially by definition and not due to a non-trivial underlying mathematical phenomenon), as explained in Fact \ref{fact:col} (iii) and (iv); but these are in general too weak to be of much help.}  

\begin{lemma}[Witsenhausen~{\cite{Witsenhausen75SIAP}}]
  \label{lemma:correlation-tensor}
  For~$i\in[n]$, let~$\rho_i$ be a probability distribution on~$U_i\times V_i$.
  Then it holds that~$
    \Cor(\rho_1\otimes\dots\otimes\rho_n) = \max_{i\in[n]}\Cor(\rho_i)
  $.
\end{lemma}
\begin{proof}
It is easy to see from the definition that maximum correlation of a distribution $\rho$ is just the second largest singular value of  the matrix $A_{u,v}=\frac{\rho(u,v)}{\sqrt{\rho(u)\rho(v)}}$.  The lemma then follows by noting that the singular values of the tensor product of two matrices are given by the pairwise products of the singular values of the original matrices.
\end{proof}

\subsection{Information theory}
Here we give a brief account of some of the concepts from Information Theory. For a more detailed introduction containing the omitted proofs, one can consult~\cite{cover2012}.

For a random variables $X$ over a domain $\mathcal X$, its entropy is defined as
\[ H(X)= \sum_{x\in \mathcal X} \Pr[X=x]  \cdot \log\left(\frac{1}{\Pr[X=x]}\right).\]
The entropy of a random variable $X$ conditioned on $Y$ is the average (according to $Y$) of entropies of random variables $X| Y=y$. More concisely, we have $H(X|Y)=H(XY)-H(Y)$.

Given random variables $X$ and $Y$ we define their mutual information via
\[ I(X;Y)= H(X)+H(Y)-H(XY)= H(X)-H(X|Y)=H(Y)-H(Y|X).\]

The main fact that we shall need regarding the mutual information is the following.
\begin{fact}[Chain rule]
For any random variables $X,Y,Z,W$ we have
\[ I(X;YZ|W)= I(X;Y|W)+ I(X;Z|Y,W).\]
\end{fact}

Finally, we will need a simple lemma which we use in our analysis.

\begin{lemma}
  \label{lemma:entropy-difference}
  Let~$X$ be a random variable uniformly distributed over the domain $\mathcal X$
  and let~$P$ be an event.
  Then, \[
    H(X\mid P)
    \ge H(X)- \log\left(\frac{1}{\Pr[ P]}\right)
  \]

\end{lemma}
We note that the above lemma does not hold in the current form when $X$ is not uniform. To see this, let $X_1,X_2$ be uniform independent random variables over $\mathcal X_1, \mathcal X_2$ respectively. Let $i\in \{1,2\}$ be also uniform and let $X=X_i$. Taking  $P$ to be the event $i=1$ and assuming that $|\mathcal X_1|\ll |\mathcal X_2|$ we see that in this case   $H(X|P)$ will be much less than $H(X)-1$ which proves the necessity of uniformity assumption.
\begin{proof}[Proof of Lemma \ref{lemma:entropy-difference}]
 We have 
 \begin{align*} 
 	H(X|P) &= \sum_{x\in \mathcal X} \Pr[X=x|P] \log\left(\frac{1}{\Pr[X=x| P]}\right)= \sum_{x\in \mathcal X} \Pr[X=x|P] \cdot \log\left(\frac{\Pr[P]}{\Pr[X=x \wedge P]}\right) \\
 	&=  \log(\Pr[P])+ \sum_{x\in \mathcal X} \Pr[X=x|P] \cdot  \log\left(\frac{1}{\Pr[X=x \wedge P]}\right) \\
 	&\geq \log(\Pr[P]) + \sum_{x\in \mathcal X} \Pr[X=x|P] \cdot \log\left(\frac{1}{\Pr[X=x]}\right) \\
 	&= \log(\Pr[P]) +  \sum_{x\in \mathcal X} \Pr[X=x|P]\cdot  \log|\mathcal X| \\
 	&= \log(\Pr[P] \cdot |\mathcal X|),
 \end{align*}
where the latter is clearly equal to $H(X)- \log(\frac{1}{\Pr[P]})$. 
\end{proof}
\subsection{Communication complexity}

We assume some basic familiarity with communication complexity~\cite{KusNis97}.
For a partial function~$f$ from~$\{0,1\}^n\times\{0,1\}^n$ to~$E$,
we denote by~$\RII(f)$ the SMP communication complexity of~$f$
without shared randomness with worst-case error probability (over the private coins of Alice and Bob) at most~$1/3$.

For~$n\ge1$, let~$\IP_n\colon\{0,1\}^n\times\{0,1\}^n\to\{0,1\}$
be the inner product function modulo $2$.
Chor and Goldreich~\cite{ChoGol88SICOMP} showed that
any communication protocol which answers~$\IP_n$
with error probability~$1/3$ on average,
even in the two-way communication model,
must have communication cost at least~$n/2-o(n)$.
The following is an immediate corollary of this.

\begin{cor} \label{theorem:inner-product}
  For~$n\ge1$, the communication complexity of~$\IP_n$
  in the SMP model with public randomness is in~$\Omega(n)$.
\end{cor}

We will use the following fact
originally proved by Newman~\cite{Newman91IPL}
and refined by Kushilevitz and Nisan~\cite[Theorem~3.14]{KusNis97}.
The statement in~\cite{KusNis97} assumes that the function
takes a Boolean value, but the proof does not use this assumption.

\begin{lemma} \label{lemma:reduce-randomness}
  Let~$0<\varepsilon<\varepsilon'<1/2$ be constants.
  Any SMP protocol with public randomness
  for a partial function from~$\{0,1\}^n\times\{0,1\}^n$ to~$E$
  with error probability at most~$\varepsilon$
  can be converted to one which uses only~$\log n+C$ bits of public randomness
  with error probability at most~$\varepsilon'$
  without changing the communication cost,
  where~$C=C(\varepsilon, \varepsilon')>0$ is a constant
  depending only on~$\varepsilon$ and~$\varepsilon'$.
\end{lemma}

\section{Usefulness of all non-product sources of shared randomness}\label{sec:equality}
In this section, we prove Theorem \ref{theorem:eq} showing an easy but nontrivial use of imperfect shared randomness in an SMP protocol for equality function. 

As a consequence of this theorem, we see that the lower bound claimed in the second part of our main theorem (Theorem \ref{thm:main}) does not hold for an arbitrary family of functions $\{g_n\}$. Thus a good choice of a function family was necessary there. 

Beyond its bearing on Theorem \ref{thm:main}, the result is interesting on its own as it explicitly shows that in some settings \emph{any form of public shared randomness} (i.e. one coming from a non-product source $\rho$) can have a significant impact on the asymptotic communication cost.

\begin{proof}[Proof of Theorem \ref{theorem:eq}]
  Because~$\rho$ is not a product distribution,
  there exist subsets~$\Lambda_a\subseteq U$ and~$\Lambda_b\subseteq V$
  such that 
  \begin{equation}
    \gamma=\Pr_{(u,v)\sim \rho} [u\in \Lambda_a\land v\in \Lambda_b]
    \ne
    \Pr_{u\sim \rho_A} [u\in \Lambda_a]\Pr_{v\sim \rho_B}[v\in \Lambda_b]=\gamma'.
    \label{eq:equality-1}
  \end{equation}

  A protocol for solving the equality is as follows.
  Alice and Bob share~$2^n$ copies of shared randomness;
  label the~$2^n$ i.i.d.~pairs of registers containing the shared randomness
  as~$(u_x,v_x)\sim\rho$ for~$x\in\{0,1\}^n$.
  Alice defines~$\alpha=1$ if~$u_x\in \Lambda_a$ and~$\alpha=0$ otherwise,
  and sends~$\alpha$ to the referee.
  Similarly, Bob defines~$\beta=1$ if~$v_y\in \Lambda_b$ and~$\beta=0$ otherwise,
  and sends~$\beta$ to the referee.
  The referee checks whether~$\alpha=\beta=1$ or not.
  If~$x=y$, this happens with probability~%
  $\gamma$,
  and otherwise it happens with probability~%
  $\gamma'$.
  The referee can tell which is the case with error probability at most~$1/3$
  by repeating this protocol for~$t
    =O(\abs{\gamma-\gamma'}^{-2})
  $ which is a constant independent of~$n$.
\end{proof}

\section{Collision complexity}
  \label{section:collision-complexity}
In the introduction we gave an informal definition and briefly mentioned some of the basic properties of collision complexity. In this section, we look at this notion in more detail and provide a more formal definition. We also give an alternate and more analytic characterisation of collision complexity in terms of a different measure of quality of shared randomness called \emph{agreement complexity}. Beside providing a different perspective on collision complexity, this alternate characterisation will be also useful technically  at various places later.



We start our work toward the formal definition of collision complexity by defining a collision protocol. 

\begin{definition}[Collision protocol]
Let $\rho$ be a probability distribution on $U\times V$.
A $p$-collision protocol
\footnote{Think of $p$ as a decreasing function of $n$:\ say, $p=1/n$.}
for~$\rho$ with domain size~$n$ is determined by two functions, $A:U^\ell\rightarrow \mathcal P[n]$ and $B:V^\ell\rightarrow \mathcal P[n]$ for some $\ell\in \NN$, such that 
\begin{equation} \label{eq:collision_condition} \forall i\in [n]\, : \qquad \Pr_{(\vec{u},\vec{v})\sim\rho^{\otimes \ell}}\left[ (\vec{u}, \vec{v}): \, i\in A(\vec{u})\cap B(\vec{v})\right] \geq p.	
\end{equation}
\end{definition}

It is sometimes convenient to think of Alice's and Bob's functions $A$ and $B$ as randomised mappings, and that is allowed:\ we always assume that $\rho$ is non-degenerate, so the players can use it to generate private randomness (it doesn't have to be hidden from the other party, as the players are always cooperating).

The \emph{complexity} of a collision protocol $\{A,B\}$ is given by the maximum output size:
\begin{equation}\label{eq:collision_objective}
\col_\rho(n,p, \{A,B\})\overset{\textrm{def}}{=}
 \max\left\{\abs{A(\vec u)}, \abs{B(\vec v)} \, :\, 
  (\vec{u},\vec{v})\in supp(\rho)\right\}.
\end{equation}

\begin{definition}[Collision complexity]
  \label{d_coll}
  Let $n\in \NN$ and $p=p(n)\in [0,1]$ be a probability parameter possibly depending on $n$. The $p$-collision complexity of a source ~$\rho$ (over a domain of size~$n$) is the minimum worst case output size necessary for any protocol achieving a collision probability of $p$ for all $i\in [n]$.
In other words,
  \[ \col_\rho(n,p)= \inf_{\ell, A, B} \col_\rho(n,p, \{A,B\}),\]
  where the infimum is over all protocols $(\ell, A,B)$ satisfying (\ref{eq:collision_condition}).
\end{definition}

The main parameter setting of interest for us here will be $p=1/n$ and as such what follows we shall let 
\begin{equation}
\col_\rho(n)\overset{\textrm{def}}{=} \col_\rho(n, 1/n).
\end{equation}

Note that the collision complexity of a non-degenerate distribution can be at most $\asO{\sqrt n}$ (achievable using private randomness via the birthday paradox).
Similarly, the rest of the properties from Fact~\ref{fact:col} easily follow from the definition.


Our second measure of correlation is \e{agreement complexity}.
It is closely related to collision complexity, but sometimes will be more convenient to work with.

\begin{definition}[Agreement complexity]
  \label{definition:agreement}
  Let~$\rho$ be a probability distribution on~$U\times V$.

  An \emph{agreement protocol} for~$\rho$
  is determined by~$\ell\in\NN$ and a pair of functions~%
  $f\colon U^\ell\to[0,1]$ and~$g\colon V^\ell\to[0,1]$.
  The \emph{cost} of this agreement protocol
  is~$\E{f(u_1,\dots,u_\ell)+g(v_1,\dots,v_\ell)}$,
  and the \emph{success probability} of this protocol
  is~$\E{f(u_1,\dots,u_\ell)g(v_1,\dots,v_\ell)}$,
  where~$(u_i,v_i)\sim\rho$ independently for all $i\in[\ell]$.

  The \emph{agreement complexity} of~$\rho$ at success probability~$p$,
  denoted by~$\agr_\rho(p)$,
  is the infimum of the cost of an agreement protocol for~$\rho$
  with success probability at least~$p$.
\end{definition}

Here Alice and Bob output one bit each (as opposed to a subset of~$[n]$ in the case of collision complexity).
The value~$f(u_1,\dots,u_\ell)$ is the probability that Alice outputs ``$1$'', given her part of shared randomness, and similarly for~$g(v_1,\dots,v_\ell)$.
The players' task is to output ``$1$'' simultaneously with probability at least~$p$,
while minimising the sum of the probabilities that each party outputs ``$1$'' -- the corresponding infimum is the agreement complexity.

\subsection{Equivalence of collision and agreement complexities and further properties}
The first lemma provides some basic parameter trade-offs.
\begin{lemma}\label{lem:parameter_manip}
For any positive integers $m,n$  and $p\in (0,1)$ we have:
 \begin{equation} \label{eq:parameters-1}
    \col_\rho(n,1-(1-p)^m) \le m\col_\rho(n,p), \qquad 
    \col_\rho(mn,p) \le m\col_\rho(n,p).
    \end{equation}
    \end{lemma}

\begin{proof}
	The first inequality follows by repeating a collision protocol~$m$ times independently
  and outputting the union of the results. For the second one, repeat the collision protocol~$m$ times independently
  to obtain~$A_1,\dots,A_m\subseteq[n]$ and~$B_1,\dots,B_m\subseteq[n]$,
  and output~$A=\{n(i-1)+j\colon i\in[m],j\in A_j\}$
  and~$B=\{n(i-1)+j\colon i\in[m],j\in B_j\}$.
  This gives a collision protocol with domain size~$mn$,
  collision probability~$p$, and output size at most~$m\col_\rho(n,p)$.
\end{proof}

As we mentioned collision and agreement complexities are intimately related. This is captured by the following lemma.
\begin{lemma} \label{lemma:col-agr}
  For a bipartite distribution~$\rho$, $n\ge1$, and~$0<p<1$,
  it holds that~$
    \col_\rho(n,p) = \Theta(\max\{1,n\agr_\rho(p)\})
  $,
  where the constant in the asymptotic notation
  does not depend on~$\rho$, $n$, or~$p$.
\end{lemma}
The proof is based on the following idea:
If~$i\in[n]$, a collision protocol~$(\ell,A,B)$ of domain size~$n$
can be converted to an agreement protocol
where Alice outputs~$1$ if~$i\in A$ and Bob outputs~$1$ if~$i\in B$.
By choosing the ``best'' value of~$i$,
we obtain the claimed upper bound on the agreement complexity.
For the conversion in the opposite direction,
Alice and Bob repeat the agreement protocol~$n$ times in parallel
to decide whether their output should contain each~$i\in[n]$.
By the Chernoff bound, this provides the necessary upper bound
on the collision complexity after an application of Lemma \ref{lem:parameter_manip}. The formal proof is as follows:
\begin{proof}
Let~$(\ell,A,B)$ be a collision protocol for~$\rho$
with domain size~$n$, output size~$\col_\rho(n,p)$,
and collision probability at least~$p$.
For~$i\in[n]$, let~$f_i(u)=\PR[i\in A(u)]$ for~$\vec u\in U^\ell$
and let~$g_i(\vec v)=\PR[i\in B(\vec v)]$ for~$\vec v\in V^\ell$.
Because the collision probability of~$(\ell,A,B)$ is at least~$p$,
it holds that~$\E{f_i(\vec u)g_i(\vec v)}\ge p$ for every~$i\in[n]$.
Because The output size being is~$\col_\rho(n)$,
it holds that~$\abs{A(\vec u)}\le \col_\rho(n,p)$ and~$\abs{B(\vec v)}\le \col_\rho(n,p)$
for all~$\vec u\in U^\ell$ and~$\vec v\in V^\ell$,
and in particular, it implies that
\[
  \E{\abs{A(\vec u)}+\abs{B(\vec v)}}
  =
  \sum_{i\in[n]}\E{f_i(\vec u)+g_i(\vec v)}
  \le 2\col_\rho(n,p).
\]
Therefore, there exists~$i^*\in[n]$ such that
\[
  \E{f_{i^*}(\vec u)+g_{i^*}(\vec v)}
  \le \frac{2\col_\rho(n,p)}{n}.
\]
Then the pair~$(f_{i^*},g_{i^*})$
is an agreement protocol with success probability at least~$p$
and cost at most~$2\col_\rho(n,p)/n$.
Therefore, we have that~$n\agr_\rho(p)\le2\col_\rho(n,p)$.

Next we prove that~$\col_\rho(n,p) \le 9n\agr_\rho(p)+48$.
Because~$\agr_\rho(p)\ge2p$, this inequality is trivial if~$p\ge1/2$.
For the rest of the proof, assume~$p<1/2$,
and we will prove the inequality in two steps:
\begin{align}
  \col_\rho(n,p/2) &\le 3n\agr_\rho(p)+16,
  \label{eq:col-agr-1}
  \\
  \col_\rho(n,p) &\le 3\col_\rho(n,p/2).
  \label{eq:col-agr-2}
\end{align}

Let~$K>\agr_\rho(p)$.
By the definition of the agreement complexity,
there exists an agreement protocol~$(\ell,f,g)$ for~$\rho$
with success probability at least~$p$ and cost less than~$c$.
Let~$a=\E_{\vec u\sim\rho_{U^\ell}}{f(\vec u)}$ and~$b=\E_{\vec v\sim\rho_{V^\ell}}{g(\vec v)}$,
and we have that~$a<K$ and~$b<K$.
Let~$T=3nK+16$.
Consider the following collision protocol.
Alice interprets her random input as~$(\vec u_1,\dots,\vec u_n)\in(U^\ell)^n$,
and Bob interprets his random input as~$(\vec v_1,\dots,\vec v_n)\in(V^\ell)^n$,
so that~$(\vec u_i,\vec v_i)$ is distributed according to~$\rho^{\otimes \ell}$ for each~$i\in[n]$
and~$n$ pairs~$(\vec u_1,\vec v_1),\dots,(\vec u_n,\vec v_n)$ are mutually independent.
Alice constructs sets~$\tilde{A},A\subseteq[n]$ as follows.
For each~$i\in[n]$,
$i$ belongs to~$\tilde{A}$ with probability~$f(\vec u_i)$,
independently of everything else.
If~$\abs{\tilde{A}}\le T$, then she defines~$A=\tilde{A}$;
otherwise~$A=\varnothing$.
She outputs set~$A$.
Bob constructs sets~$\tilde{B}$ and~$B$ in the analogous way,
and outputs~$B$.
It is clear that the output size of this collision protocol is at most~$t$.

Let~$i\in[n]$.
By the Chernoff bound, it holds that
\begin{align*}
  \Pr[A\ne\tilde{A} \mid i\in\tilde{A}\cap\tilde{B}]
  &=
  \Pr[\abs{\tilde{A}}>T \mid i\in\tilde{A}\cap\tilde{B}]
  \\
  &=
  \Pr[\abs{\tilde{A}\setminus\{i\}}>t-1 \mid i\in\tilde{A}\cap\tilde{B}]
  \\
  &\le
  \frac{1}{e^{a(n-1)}}\cdot\left(\frac{ea(n-1)}{T-1}\right)^{T-1}
  \\
  &<
  \left(\frac{e}{3}\right)^{T-1}
  <
  \frac14,
\end{align*}
where the last inequality follows from~$(e/3)^{15}<1/4$.
Similarly, it holds that~$\Pr[B\ne\tilde{B}\mid i\in\tilde{A}\cap\tilde{B}]<1/4$.
By union bound, we have that
\[
  \Pr[A=\tilde{A}\land B=\tilde{B} \mid i\in\tilde{A}\cap\tilde{B}]>\frac12.
\]
Therefore,
\[
  \Pr[i\in A\cap B]
  =
  \Pr[i\in\tilde{A}\cap\tilde{B}]
    \Pr[A=\tilde{A}\land B=\tilde{B} \mid i\in\tilde{A}\cap\tilde{B}]
  >
  \frac{p}{2},
\]
meaning that the collision probability of this collision protocol
is greater than~$p/2$.
This implies that~$\col_\rho(n,p/2) \le 3nK+16$.
Because this holds with any~$K>\agr_\rho(p)$,
inequality~(\ref{eq:col-agr-1}) follows.

Inequality~(\ref{eq:col-agr-2}) follows from~(\ref{eq:parameters-1})
by noting that~$1-(1-p/2)^3>p$ because~$0<p<1/2$.
\end{proof}
\section{Simulation of  $\RII_{\it perf}$  protocols in $\RII_{\it \rho}$} 
Given the technical tools we developed in the previous section we are ready  to prove the first (easy) part of Theorem \ref{thm:main}. 

\begin{proposition} \label{lemma:repl}
  Let~$\rho$ be a bipartite distribution. Then for any (possibly partial) function $f$ over~$\{0,1\}^n\times\{0,1\}^n$ we have 

    \[
      \RII_\rho(f) =  O(\col_\rho(n)\, (\mathtt{\RII_{\it perf}}(f)+\log n)\,).
    \]
\end{proposition}

To prove the lemma, we first need a symmetrisation claim which shows that at a small multiplicative cost one can always turn a collision protocol into one which treats all elements $i\in [n]$ the same way. More formally, we have the following:

\begin{claim} \label{claim:uniform-collision-protocol}
  Let~$0<s<1$.
  For a bipartite distribution~$\rho$ and~$n\ge1$,
  there exists a collision protocol~$(\ell,A,B)$
  with domain size~$n$ and output size at most~$O(\col_\rho(n))$
  that satisfies the following two properties:
  \begin{enumerate}[(a)]
  \item
    $\Pr_{(u,v)\sim\rho^{\otimes \ell}}[A(u)\cap B(v)=\varnothing]\le s$.
  \item
    Conditioned on the event~$A(u)\cap B(v)\ne\varnothing$,
    selecting one element from~$A(u)\cap B(v)$ uniformly at random
    produces the uniform distribution over~$[n]$.
  \end{enumerate}
  The constant in the asymptotic notation depends on~$s$, the failure probability which is a small but fixed constant,
 however it does not depend on~$\rho$ or~$n$.
\end{claim}

\begin{proof}
  We first prove the case where~$s\ge1/e+1/2$.
  Consider the collision protocol used in the proof of inequality~(\ref{eq:col-agr-1})
  in Lemma~\ref{lemma:col-agr} with~$p=1/n$.
  This protocol satisfies condition~(b) by symmetry.
  Its output size is at most~$3n\agr_\rho(1/n)+16$,
  which is in~$O(\col_\rho(n))$ by Lemma~\ref{lemma:col-agr}.
  By the Chernoff bound, it holds that
  \[
    \Pr[A\ne\tilde{A}]
    =
    \Pr[\abs{\tilde{A}}>T]
    \le
    \frac{1}{e^{an}}\cdot\left(\frac{ean}{T}\right)^T
    <
    \left(\frac{e}{3}\right)^T
    <
    \frac14,
  \]
  and similarly~$\Pr[B\ne\tilde{B}]<1/4$.
  By union bound, it holds that
  \begin{align*}
    \Pr[A\cap B=\varnothing]
    &\le
    \Pr[\tilde{A}\cap\tilde{B}=\varnothing]
    +
    \Pr[A\ne\tilde{A}]
    +
    \Pr[B\ne\tilde{B}]
    \\
    &<
    \left(1-\frac{1}{n}\right)^n
    +\frac14+\frac14
    \\
    &<
    \frac{1}{e}+\frac12.
  \end{align*}
  This proves the case where~$s\ge1/e+1/2$.

  To prove the general case,
  repeat this protocol for~$\ceil{\log s/\log(1/e+1/2)}$ times
  and output the union of the results.
\end{proof}

\begin{proof}[Proof of Proposition~\ref{lemma:repl}]
Using randomness reduction, i.e.~Lemma~\ref{lemma:reduce-randomness},
  we know that there exists an SMP protocol for~$f$
  using at most  ~$\log n + C$ bits of public perfect randomness achieving with error probability at most~$2/5$.  Fix this protocol and set $m=\log n+C$ and $\epsilon=2/5$ and  $s=(1-2\varepsilon)/(4-4\varepsilon)$. Let~$\alpha(x,r)$ to be the message which Alice sends given input $x$ and public (perfect) random string  ~$r\in 2^{[m]}$,
  and let~$\beta(y,r)$ be the message which Bob sends
given input is~$y$ and randomness~$r$.
  Let~$\gamma(a,b)\in E$ be the referee's output
  given that the message from Alice is~$a$ and the message from Bob is~$b$.

  Consider the following SMP protocol for~$f$ with shared distribution~$\rho$.
  \begin{enumerate}
  \item
    Alice and Bob run the collision protocol
    in Claim~\ref{claim:uniform-collision-protocol}
    with domain size~$2^m$ and parameter~$s$ specified above.
    Identify the domain~$[2^m]$ with the set of~$m$-bit strings,
    and let~$A,B\subseteq\{0,1\}^m$
    be the sets that Alice and Bob obtain as a result
    of running this collision protocol.
    Alice sends the referee each string~$r$ in~$A$ together with~$\alpha(x,r)$,
    and Bob sends the referee each string~$r$ in~$B$ together with~$\beta(y,r)$.
  \item
    The referee outputs an arbitrary answer in~$E$
    if~$A\cap B=\varnothing$.
    Otherwise, he chooses an element~$r\in A\cap B$ uniformly at random,
    and he outputs~$\gamma(\alpha(x,r),\beta(y,r))$
    (which he can do, because the values~$\alpha(x,r)$ and~$\beta(y,r)$
    have been sent by Alice and Bob, respectively).
  \end{enumerate}

  The complexity of the above protocol is clearly in~$\asO{\col_\rho(2^m)(\RII_{\it perf}+m)}$. It is clear that the protocol succeeds with probability ~$(1-s)(1-\varepsilon)=(3-2\varepsilon)/4>1/2$. By amplification, we are done.
\end{proof}

\section{Characterising $\mathit{SMP}_\rho(f)$ via collision complexity: the lower bound}
  \label{s_GIP}

In the previous section we proved the first part of Theorem \ref{thm:main}. In this section,  we complete what we started there by proving the second (and the harder) part of the theorem.  The main ingredient of  the proof is a lower bound on the communication complexity of  $\GIP_n$ functions (introduced in Definition~\ref{d_gip}) which is the content of Theorem \ref{t_GIP}. 

\begin{theorem}
  \label{t_GIP}
  For any bipartite distribution~$\rho$, it holds that~$
    \RII_\rho(\GIP_n) = \Omega(\col_\rho(n))
  $,
  where the constants in the asymptotic notations do not depend on~$\rho$ or~$n$.
\end{theorem}
The proof of the above Theorem is rather involved and requires several steps. We shall give an exposition of  the main ideas behind the proof  in the next subsection, and then carry out the proof in details in the remaining three subsections. Before delving into that, however, let us observe that this  Theorem \ref{t_GIP} together with two easy propositions imply the second half of Theorem \ref{thm:main}. 

\begin{proposition} \label{proposition:gip-upper-bound}
  For a bipartite distribution~$\rho$, it holds that~$
    \RII_\rho(\GIP_n) = O(\col_\rho(n)\log n)
  $.
\end{proposition}

\begin{proof}
  Note that~$\GIP_n$ has a straightforward SMP protocol with public randomness
  with cost~$O(\log n)$ and error probability at most~$1/3$:
  Alice and Bob choose a common index~$i\in[n]$ uniformly at random,
  Alice sends~$x_i$ to the referee, and Bob sends~$y_i$ to the referee.
  The referee simply answers~$x_i\cdot y_i$.
  Because this protocol uses~$\log n$ bits of public randomness,
  Lemma~\ref{lemma:repl}~(i) implies the claim.
\end{proof}

\begin{proposition} \label{proposition:gip-lower-bound-log}
  For a bipartite distribution~$\rho$, it holds that~$
    \RII_\rho(\GIP_n) = \Omega(\log n)
  $.
\end{proposition}

\begin{proof}
  Note that~$\IP_m$ can be reduced to~$\GIP_{n,m}$
  by repeating the input vector~$n$ times,
  and in particular~$\IP_{8\log n}$ can be reduced to~$\GIP_n$.
  By Corollary~\ref{theorem:inner-product}, this implies the claim.
\end{proof}
From the above we get the following corollary which was our main goal.
\begin{cor}[Second half of Theorem \ref{thm:main}]
  \label{t_diff}
  There exists a family of partial functions~$f_N$ from~$\{0,1\}^N\times\{0,1\}^N$
  to~$\{0,1\}$ such that for any bipartite distribution~$\rho$, it holds that
  \[
    \RII_\rho(f_N)
    = \begin{cases}
    O\left(\col_\rho\left(\frac{N}{\log N}\right)\log N\right)\qquad  &(\text{upper bound}) \\
    \Omega\left(\max\left\{\log N,\col_\rho\left(\frac{N}{\log N}\right)\right\}\right)  \qquad &(\text{lower bound})

    \end{cases}
    \, ,
  \]
  where the constants in the asymptotic notations do not depend on~$\rho$ or~$N$.
\end{cor}
The reason for appearance of $\frac{N}{\log N}$ terms in the above  statement is that the input size to $\GIP_n$ is $N=nm= 8n\log n$ which essentially is the main source of the poly-logarithmic gap between the upper and lower bounds in the above. However, note that the upper and lower bounds in Corollary \ref{t_diff} are still quite close, because thanks to  inequality~(\ref{eq:parameters-1}) we have ~$\col_\rho(N/\log N)=\Omega(\col_\rho(N)/\log N)$---which means that  the result is tight up to a factor of~$O(\log^2 N)$ for every bipartite distribution~$\rho$.

\subsection{The proof outline of Theorem \ref{t_GIP}}
The proof of Theorem \ref{t_GIP} goes roughly as follows. Consider a hypothetical protocol for~$\GIP_n$ with complexity $o(\col_{\rho}(n))$. We will consider the behaviour of the protocol under the uniformly random input distribution---ignoring the promise on the input of~$\GIP_n$ for the moment. Because of the low communication complexity of the protocol we can argue (using a ``rounding argument" which extracts a collision protocol from a $\GIP_n$ protocol by looking at the influential coordinates) that for a typical $i\in[n]$ the referee can guess the value of $x_i\dt y_i$ correctly only with probability at most $1/2+O(1/\poly(n))$ where the $\poly(n)$ here is a polynomial with \emph{super-linear growth}.

Next, we will use the repeat the above argument inductively $\Omega(n)$ times to conclude that there exists a set $L\subseteq[n]$ of size at least $2n/3$, such that the referee can distinguish the case when $x_i\dt y_i=0$ for all $i\in L$ from the case when $x_i\dt y_i=1$ for all $i\in L$  only with probability $1/2+o(1)$, where here we crucially rely on the fact the bias of the protocol in the previous step was only $o(1/n)$. Finally, if we define the input distribution $\mu$ to be uniform over $\set{(x,y)\colon \forall i\in L:x_i=y_i}\cup\set{(x,y)\colon\forall i\in L:x_i\ne y_i}$, then each element in the support of $\mu$ is a valid input to~$\GIP_n$, but the protocol under consideration makes a mistake with respect to $\mu$ with large probability which is in contradiction with the initial hypothesis. 

The above outline  essentially  contains a full description of our argument except for one subtle but important technical point: It turns out that as the hybrid argument progresses, turning the distribution of $(x_i,y_i)$ from uniform to $x_i\cdot y_i=0$ (say) for many coordinates $i$, the players may in principle be able to use the extra shared randomness offered  by conditioning on the previous rounds to their advantage. All of this essentially amounts to the fact, that in the midst of the hybrid argument the players instead of just having access  to $\rho$, instead have access to the more powerful source $\rho\otimes \sigma_0$ (or $\rho\otimes \sigma_1$)---where $\sigma_b$ denotes the uniform distribution on $(x^*,y^*)\in \F_2^m\times \F_2^m$ conditioned on $x^*\cdot y^*=b$. In fact, we do not know any way to show that the players cannot take advantage of the newly introduced shared randomness caused by conditioning, rather we use some ideas related to maximum correlation to show that, for our specific choice of parameter  $m=8\log n$, one has
\begin{equation}\label{eq:collision_extra_SR}
	 \col_{\rho\otimes \sigma_0}(n)=\Omega(\col_{\rho}(n)), \qquad \col_{\rho\otimes \sigma_1}(n)=\Omega(\col_{\rho}(n)),
\end{equation}
which suffices for our application. 

We note that similar issues to the one above, i.e.~the  emergence of additional undesired correlations between parties as a result of conditioning, occur in many different contexts in the theoretical computer science literature, for example in the context of parallel repetition \cite{hol07,raz98} and  also in many other places in communication complexity \cite{braverman2013, jain2011}. However, the method we use here to handle this issue is quite different from the prior work, and hence, could potentially be of interest for those applications as well.

\subsection{Hardness of $\RAIP_n$: the inductive step}

In our analysis, it is convenient to use a slight variant of the standard SMP model, which we call the \emph{pseudo-SMP} model, with the following properties:
\begin{itemize}
 \item The referee can receive his own portion of input.
 \item If shared randomness is available, the referee can see both Alice's and Bob's part of the shared distribution.
\end{itemize}

The crucial condition here is the latter, i.e.~the fact that Alice and Bob's shared randomness (but not inputs) is visible to the referee. Note that in  SMP, the players and the referee cooperate in order to solve the problem, giving the referee the above extra power only makes the model stronger; thus, any lower-bound established for pseudo-SMP immediately implies a corresponding lower bound for the usual SMP model (possibly with some input the referee).\footnote{The reason we do this here is that to prove the lower bound in Lemma \ref{lemma:RAIP} (even for just basic SMP), we seem to need a lower bound  for $\RAd$ against the slightly stronger pseudo-SMP model.}

For our lower bound it is natural to introduce some auxiliary communication tasks. First consider the following general ``guessing" information theoretic task, which is a two-player variant of a problem considered previously by Gavinsky, Kempe, Regev, and de~Wolf~\cite{GavKemRegWol09SICOMP}:

\begin{definition}[Random access problem]
\label{d_b-dist}
In $\RAd$ problem, Alice receives  $X\in\01^n$, Bob receives $Y\in\01^n$, and the referee receives $I\in[n]$.  A pseudo-SMP protocol solves $\RAd$ problem successfully if the referee's output $Z\in\01^2\cup\set{\bot}$ satisfies
\[ \PR_{X,Y\sim\U}[Z=(X_i,Y_i)\mid I=i,\;Z\neq\bot]\ge1-\delta, \qquad \forall \, i\in [n]. \]
\end{definition}

Next we introduce the Random Access Inner Product Problem, which is slightly less information theoretic, and closer to $\GIP_n$.

\begin{definition}[Random access inner product problem]
\label{d_gip'}
Let $n,m\in\NN$ and suppose $x_i,y_i\in\F_2^m$ for $i\in [n]$
Denote $x=(x_1\dc x_n)$ and $y=(y_1\dc y_n)$.
Then $\RAIP_{n,m}(x,y,i)\deq\ x_i\dt y_i$.
\end{definition}
 We will write~$\RAIP_n$ to address $\RAIP_{n,8\log n}$ and let $I=i$ denote the event that the third input to $\RAIP_{n}(\cdot, \cdot, I)$ is fixed to $i$.  

The main purpose of this subsection is to prove the next lemma, which states that~$\RAIP_n$ is a hard problem; this is essentially all we need in the next subsection to run our inductive argument for  lower bounding $\RII_\rho(\GIP_n)$  there. 
\begin{lemma}
  \label{lemma:RAIP}
  Let ~$\Cl P$ be a pseudo-SMP protocol for~$\RAIP_n$
 using source~$\rho$ and communication cost~$CC_{\mathcal P}$ such that
  $CC_{\mathcal P}+\log n\leq 41\cdot \col_\rho(n)$. Then there exists $i\in [n]$ such that 
 \begin{equation}
   \gamma_i \leq \frac{1}{2}+O\left(\frac{CC_{\mathcal P}+\log n}{ n\col_\rho(n)}\right)+o\left(\frac{1}{n}\right),\footnote{The implied constants in $O(\cdot)$ and $o(\cdot)$ here are universal and, in particular, do not depend on $\rho$ or $n$. }
    \end{equation}
    where $\gamma_i$ is the probability of success of $\mathcal P$ given $I=i$ and uniformly random inputs  $x,y \in (\F_2^m)^n$. 
    \end{lemma}

To understand the statement of this lemma,
consider the case where~$\col_\rho(n)=\omega(\log n)$.
Then the lemma claims that if the cost~$CC_{\mathcal P}$
of a communication protocol~$\mathcal P$ for~$\RAIP_n$ is too small,
i.e.,~$CC_{\mathcal P}= o(\col_\rho(n))$,
then there is a coordinate~$i\in[n]$ such that~$\mathcal P$ answers~$x_i\cdot y_i$
correctly with probability at most~$1/2+o(1/n)$.

To prove Lemma~\ref{lemma:RAIP}, we first show that~$\RAd$ is a hard problem and then reduce~$\RAd$ to~$\RAIP_n$, thus establishing that the latter is also hard. More precisely, we prove the following. 

\begin{lemma}
\label{l_RAd}
There exists a constant~$\delta>0$ such that  the following holds.
Let~$\Cl P$ be a pseudo-SMP protocol for $\RAd$ using source~$\rho$ and communication cost~$CC_{\mathcal P}$ such that $36(CC_{\mathcal P}+\log n)\le \col_\rho(n)$. Then there exists~$i\in[n]$ such that the probability of referee outputting anything beside $\bot$ satisfies 
\[ \Pr_{x,y}[ \mathcal{P}(x, y, i)\neq \bot]= O\left(\frac{CC_{\mathcal P}+\log n}{n\col_\rho(n)}\right)+ o\left(\frac{1}{n}\right),\]
where $x,y\in \{0,1\}^n$ are chosen uniformly at random.
\end{lemma}

\begin{proof}
 The main idea is to use the protocol $\mathcal P$ for $\RAd$ to produce  a good collision protocol (over domain of size $n$)  by considering the influential coordinates of $\mathcal P$. Following the above idea, it turns out the collision probability of the protocol we obtain will be directly related to probability  that the referee did not output $\bot$ in the original $\RAd$ protocol which implies the desired result. The details are as follows.
 
 Let $(\Ra,\Rb)\sim\rho^{\otimes \ell}$ denote the shared randomness between Alice and Bob.
  Let~$\Ma$ and $\Mb$ be the random variables
  representing Alice and Bob's messages  to the referee. We also let $Z$ denote the referee's output.

 Define random variables~$L_A$ and~$L_B$ as a  function of $\Ra=r_A$, $\Ma=m_a$, $\Rb= r_B$ and $\Mb= m_b$ by 
  \begin{align*}
    L_A&=\left\{i\in[n]\colon H(X_i|\Ra=r_A,\Ma=m_a)<\frac12\right\}, \\
    L_B&=\left\{i\in[n]\colon H(Y_i|\Rb=r_B,\Mb=m_B)<\frac12\right\}.
  \end{align*}
  In other words, $L_A$ (resp.~ $L_B$) consists of the coordinates of Alice's inputs  (resp.~Bob's inputs) where the referee, who sees both of $r_A, m_A$ (resp.~$r_B, m_B$) by definition of pseudo-SMP model\footnote{Recall that in the pseudo-SMP model, 
  the referee sees all the shared randomness 
  as well as the messages from Alice and Bob.} have substantial information about. 

  Since we are assuming $M_A$ and $M_B$ are small, we expect that typically the sets $L_A$ and $L_B$ are also small. More formally, we show that   for any~$\ra\in U^\ell$
  \begin{equation}
    \PR[\abs{L_A}\ge t\mid\Ra=\ra]<\cdot2^{CC_{\mathcal P}-t/2}.
    \label{eq:RA-3}
  \end{equation}
  where the probability is over a random choice of $x\in \{0,1\}^n$. By the definition of~$L_A$ we have 
  \begin{align*}
    &
    H(X\mid\abs{L_A}\ge t, \Ra=\ra)
    \\
    &\le
    H(M_A\mid\abs{L_A}\ge t,\, \Ra=\ra)
    +H(X\mid M_A,\;\abs{L_A}\ge t, \,\Ra=\ra)
    \\
    &=
    CC_{\mathcal P}+ \sum_{i=1}^n H(X_i | X_{<i}, M_A , \, |L_A|\geq t,\, \Ra=\ra)
    \\ 
    & \leq CC_{\mathcal P} + \sum_{i=1}^n H(X_i | M_A , \, |L_A| \geq t, \, \Ra=\ra)
    \\
    & \leq CC_{\mathcal P}+ n -\frac{t}{2}.
  \end{align*}
Now we are ready to describe our collision protocol $T=2C+4\log n$
  By Lemma~\ref{lemma:entropy-difference},
  this implies inequality~(\ref{eq:RA-3}).

  Because inequality~(\ref{eq:RA-3}) holds for every choice of~$\ra$,
  it holds without conditioning, i.e.
  \[ \PR[\abs{L_A}\ge t]<2^{CC_{\mathcal P} -t/2}.\]
  
  Setting $T=2C+4\log n$ and using the symmetric argument for Bob, we obtain that
  \begin{equation}
    \PR[\abs{L_A}\ge T\vee\abs{L_B}\ge T]
    <
    \frac{2}{n^2}.
    \label{m_Lshort}
  \end{equation}
  Recall that our goal was to ``extract" from $\mathcal P$ a \emph{collision protocol} $\mathcal P_{col}$. This is given as follows: In $\mathcal P_{col}$ Alice and Bob pick random inputs $X, Y\sim\U_{\01^n}$ and respective half of the shared randomness $(R_A, R_B)\sim \rho^{\otimes \ell}$, and  sample~$\Ma$ and $\Mb$ by simulating the action of the players in $\mathcal P$.  The players can each locally compute as $L_A, L_B$ as these are a deterministic functions of $\Ma\Ra$ and  $\Mb\Rb$ respectively. The output of the protocol is  $L_A',L'_B\subseteq [n]$, where 
  \[ L'_A  \deq \begin{cases} L_A \qquad &|L_A|<T \\
 \emptyset \qquad & |L_A|\geq T	 
 \end{cases}, \qquad L'_B  \deq \begin{cases} L_B \qquad &|L_B|<T \\
 \emptyset \qquad &|L_B|\geq T
 \end{cases}.\]
\emph{The analysis of collision protocol.}   Let~$\alpha=\min_{i\in[n]}\PR[i\in L_A\cap L_B]$. From~(\ref{m_Lshort}) it follows that $\forall i\in[n]:\PR[i\in L_A'\cap L_B']>\alpha -\frac{2}{n^{2}}$.

  This means that
  \[
    \col_\rho\left(n,\alpha-\frac{2}{n^{2}}\right)
    \le T.
  \]
  By inequality~(\ref{eq:parameters-1}), this implies that
  \[
    \col_\rho(n)
    \le
    T\cdot \ceil*{\frac{-\ln(1-1/n)}{-\ln(1-\alpha+2/n^{2})}}
    \le
   T\cdot \left(1+\frac{2/n}{\alpha-2/n^{2}}\right),
  \]
  where the second inequality follows
  from~$-\ln(1-x)\le2x$ for~$0<x\le1/2$ and~$-\ln(1-y)\ge y$ for~$y<1$.
  If~$T\le \col_\rho(n)$, then this implies that
  \begin{equation}\label{eq:first_part_analysis}
    \alpha\le\frac{4}{n}\cdot\frac{T}{\col_\rho(n)}+\frac{2}{n^{2}}.
  \end{equation}
 Next we need to related $\alpha$ to the probability that the referee outputs $\bot$ in $\mathcal P$. Combined with (\ref{eq:first_part_analysis}) this will finish the proof. Suppose $\delta<\frac{1}{41}$. Since ~$\mathcal{P}$ is a correct protocol for~$\RAd$
  and since ~$h(x)=-x\log x- (1-x)\log(1-x)<1/6$ for~$0\le x\le\delta$, we have 
  \[
    H(X_i\mid R_AR_BM_AM_B,I=i, Z\neq \bot)
    <\frac16.
  \]
Since~$X_i$ and~$R_BM_B$ are independent
  conditioned on~$R_A$ and~$M_A$, we immediately see  that
  \[
    \E_{r_A, m_B}{H(X_i| \Ra=r_A,\Ma=m_A)\mid I=i\, \, Z\neq\neg\bot}<\frac16,
  \]
  and so by Markov's inequality we see that 
  \[
    \PR[i\in L_A\mid I=i, Z\neq \bot]
    \ge2/3.
  \]
  Using the similar argument for Bob  and  the union bound, we obtain that
  $\PR[i\in L_A\cap L_B\mid I=i,\,  Z\neq \bot]\ge1/3$. Thus, for all $i\in [n]$ we have 
  \begin{align*}
  	\PR[  I=i,\,  Z\neq \bot] &\leq  3\cdot \PR[ i\in L_A\cap L_B ,\,  I=i,\,  Z\neq \bot]  \\
  	& \leq \PR[ i\in L_A\cap L_B ,\,  I=i] \\
  	&=  \Pr[ i\in L_A\cap L_B] \cdot \Pr[I=i],
  \end{align*}

  which after dividing by $\Pr[I=i]=\frac{1}{n}$  gives us
  \begin{equation}
    \PR[Z\neq \bot\mid I=i]
    \le
    3\PR[i\in L_A\cap L_B]=3\alpha.
    \label{eq:RA-2}
  \end{equation}
Combining~(\ref{eq:RA-2}) and (\ref{eq:first_part_analysis}) we have 
  \[
    \PR[Z\neq \bot\mid I=i]
    \le 3\alpha 
    = O\left(\frac{CC_{\mathcal P}+\log n}{n\cdot \col_\rho(n)}+\frac{1}{n^{2}}\right),
  \]
  which finishes the proof.
\end{proof}

It remains to see that $\RAd$ can be reduced to~$\RAIP_n$.  Here is where the fact that our lower bound for $\RAd$ in Lemma \ref{l_RAd} is against the stronger pseudo-SMP model comes handy as our reduction from $\RAd$ to $\RAIP_n$ only produces a pseudo-SMP $\RAd$ protocol---even if the original $\RAIP_n$ protocol was a standard SMP protocol. Also, here is one place in the argument where the fact that the source $\rho$ is non-degenerate becomes important because by the definition of pseudo-SMP and the fact that $\rho$ is non-degenerate,  we can assume Alice and the referee (and symmetrically also Bob and the referee) share as much shared randomness as desired.

\begin{lemma}
\label{l_GIPp}
Let~$\Cl P$ be a pseudo-SMP protocol
for~$\RAIP_n$ using shared distribution~$\rho$, and define $\gamma_i$ be the probability, over the uniform
input distribution over ~$(x,y)\in(\F_2^m)^n\times(\F_2^m)^n$, that~$\Cl P$ gives the correct answer conditioned on $I=i$. 
Then, there exists a pseudo-SMP protocol $\Cl P'$
with the same communication cost
and solves~$\RAd$ with $\delta=o(1/n)$ such that 
\begin{equation}\label{eq:output_bad_reduction}
	\Pr[ \mathcal P'(X, Y, I)\neq \bot]\geq 2\gamma_I-1-o\left(\frac{1}{n}\right).
\end{equation}

\end{lemma}
In the proof we denote Alice's  and Bob's messages to the referee in Protocol~$\mathcal{P}$ by $M_A=M_A(a_1,\dots,a_n,R_A)$  $M_B=M_B(b_1,\dots,b_n,R_B)$. The referee's answer will be denoted by $Z=Z(i,R_A,R_B,M_A,M_B)\in\{0,1\}^2\cup\{\bot\}$.
\begin{proof}
Let us describe a protocol $\Cl P'$ for~$\RAd$ (derived from protocol $\mathcal P$ for $\RAIP_n$):
  \begin{enumerate}
  \item
    Alice and Bob receive~$x,y\in\01^n$ and their part of shared randomness, $R_A, R_B$. The referee sees both~$R_A$ and~$R_B$.
  \item
    Alice and the referee use the randomness they share (independent of~$R_A$)
    to choose uniformly at random $2n$ values $a_{j,t}\in\F_2^m$,
    for $(j,t)\in[n]\in \{0,1\}$.
    Similarly, Bob and the referee use their shared randomness
    to choose  $b_{j,t}\in\F_2^m$
    for $(j,t)\in[n]\times \{0,1\}$ at random.
  \item
    Alice and Bob send the referee messages~$\Ma=\Ma(a_{1,x_1},\dots,a_{n,x_n},R_A)$ and $\Mb=\Mb(b_{1,y_1},\dots,b_{n,y_n},R_B)$ where $\Ma$ and $\Mb$ are as in $\mathcal P$.
  \item
    For each choice of~$s_1,\dots,s_n\in\F_2^m$,
    the referee computes
    \[
      \alpha_{s_1,\dots,s_n} = Z(\Ra,\Rb,\Ma,\Mb(s_1,\dots,s_n,\Rb)) \in \{0,1\},
    \]
   and also $\beta_{s_1,\dots,s_n}$ defined similarly.

  \item
    The referee computes~$z_A,z_B\in\{0,1,\bot\}$ defined as follows.
    If~$\alpha_{s_1,\dots,s_n}$ and~$a_{i,0}\cdot s_i$ agree
    for more than~$2^{mn}(1/2+1/2^{m/3})$ choices
    of~$(s_1,\dots,s_n)\in(\F_2^m)^n$,
    then let~$z_A=0$;
    otherwise, if~$\alpha_{s_1,\dots,s_n}$ and~$a_{i,1}\cdot s_i$ agree
    for more than~$2^{mn}(1/2+1/2^{m/3})$ choices
    of~$(s_1,\dots,s_n)\in(\F_2^m)^n$,
    then let~$z_A=1$; if neither holds, let~$z_A=\bot$.
    Define~$z_B$ analogously.
    If either $z_A$ or $z_B$ are $\bot$  the referee outputs $Z=\bot$ and otherwise he outputs $Z=(z_A,z_B)$. 
 
  \end{enumerate}
  
  First we analyse the probability that~$\mathcal{P}'$ outputs~$\bot$ to confirm  (\ref{eq:output_bad_reduction}).
  Fix~$i,R_A,R_B,x_1,\dots,x_n,y_1,\dots,y_n$
  and consider~$a_{j,t}$ and~$b_{j,t}$ as the random variables.
  Note that if~$\mathcal{P}'$ outputs~$\bot$,
  then one of the following events has to occur:
  \begin{align*}
    \bm{e}_A
    &\colon
    \#\{(s_1,\dots,s_n)\in(\F_2^m)^n\colon
    \alpha_{s_1,\dots,s_n}=a_{i,x_i}\cdot s_i\}
    \le
    2^{mn}\left(\frac12+\frac{1}{2^{m/3}}\right),
    \\
    \bm{e}_B
    &\colon
    \#\{(s_1,\dots,s_n)\in(\F_2^m)^n\colon
    \beta_{s_1,\dots,s_n}=b_{i,y_i}\cdot s_i\}
    \le
    2^{mn}\left(\frac12+\frac{1}{2^{m/3}}\right).
  \end{align*}
  Let~$\bm{e}=\bm{e}_1\lor\bm{e}_1$.
  Let $\bm{error}$ denote the event that the referee from~$\Cl P$
   given~$\Ra$, $\Rb$, $\Ma$, $\Mb$ and~$i$
  would give a  different answer from $a_{i,x_i}\tm b_{i,y_i}$.
  Then by definition
  $\PR[\bm{error}\mid I=i]=1-\gamma_i$.
  On the other hand,
  $\PR[\bm{error}\mid \bm e\land I=i]\ge1/2-2^{-m/3}$,
  and therefore
  \[
   \PR[\bm e\mid I=i]
   =\frac{\PR[\bm{error\land\bm e}\mid I=i]}{\PR[\bm{error}\mid\bm e\land I=i]}
   \le\frac{\PR[\bm{error}\mid I=i]}{\PR[\bm{error}\mid \bm e\land I=i]}
   \le\frac{1-\gamma_i}{1/2-2^{-m/3}}
   \le2-2\gamma_i+2^{2-m/3},
  \]
  where the last inequality follows from our assumption that $m\ge8$.
  Therefore, the probability that $\Cl P'$ returns an answer different from ``$\bot$'' when $I=i$ is at least $2\gamma_i-1-2^{2-m/3}=2\gamma_i-1-o(1)$. So it just remains to show that condition on being different from $\bot$, the output of $\mathcal P'$ is $(x_i,y_i)$ at least $1-o(1)$. To do so we employ a simple Fourier analytic (orthogonality) argument.
  
Let $h:\F_2^m\rightarrow [-1,1]$ be 
\[ h(y)= \Ex_{s_1, s_2, \ldots, s_{i-1}, s_{i+1}, s_n}  \left[\alpha_{s_1,\ldots, s_{i-1}, y, s_{i+1}, \ldots, s_n}\right],  \]
and define the random variable \texttt{BAD}  by
\[ \texttt{BAD}=\left\{z\in \F_2^m \, : \Ex_{y} h(y)\cdot (-1)^{z\cdot y}> 2^{1-m/3}\right\}.\]
Moreover, note that although \texttt{BAD} depends on $R_A, R_B, M_A, a_{1,x_1}, a_{2, x_2}, \ldots, a_{n,x_n}$, it is crucially independent of the choice of  $a_{1, 1-x_1}, a_{2, 1-x_2}, \ldots, a_{n, 1-x_n}$. On the other hand, the probability that  $z_A= a_{i,1-x_i}$, which corresponds to the referee's error in identifying Alice's input, is upper bounded by $\Pr[ a_{i, 1-x_i}\in \texttt{BAD}]$ which is equal to $\frac{|\texttt{BAD}|}{2^m}$ because $a_{i,1-x_i}$ was chosen uniformly at random independently of \texttt{BAD}.

Finally, by basic Fourier analysis (i.e.~the orthogonality of the characters $y\mapsto (-1)^{z\cdot y}$) it follows that 
\[ \frac{|\texttt{BAD}|}{2^m}\leq  2^{1-m/3}=o(1).\]
This means that condition on $z_A\neq \bot$, the referee with probability $1-o(1)$ predicts the correct $a_{i,x_i}$, and hence the correct $x_i\in \{0,1\}$, as Alice's input. Repeating the same argument for Bob the result follows.

\end{proof}

\begin{proof}[Proof of Lemma~\ref{lemma:RAIP}]
 This follows immediately from Lemma~\ref{l_RAd} and \ref{l_GIPp}.
\end{proof}

\subsection{Hardness of gapped inner product}

The hardness of~$\RAIP_n$ can be used as an inductive step to prove the desired lower bound for~$\GIP_n$ via a hybrid argument. Before carrying this out, let us establish some useful notations.  

For a bit $b\in \{0,1\}$ and a positive integer $m$ (as usual here $m=8\log n$), we define $\sigma_{m,b}$ to be the uniform distribution
  over~$\{(u,v)\in\F_2^m\times \F_2^m\colon u\cdot v=b\}$.  
  
   A slight (and ultimately insignificant) issue that arises in our proof is the asymmetry of $b=0$ and $b=1$ cases, i.e.~the fact that it is slightly more probable for the inner product of two random $x,y\in \F_2^m$ to be $0$ than $1$.\footnote{Note that it does not seem possible to  immediately dispose of this asymmetry because Lemma \ref{lemma:RAIP},  the basis of our induction, relies on the uniform distribution on inputs which has this asymmetry inherent with it. It is quite conceivable that we could have established a variant of Lemma  \ref{lemma:RAIP} with a modified distribution eliminating the asymmetry, but it is likely that this would have just amounted to a reshuffling of the argument.} To make the notation related to this issue simpler, we let $\gamma_b$ be the probability of two randomly chosen $x,y\in \F_2^m$ satisfy $x\cdot y=b$; that is, $\gamma_0=\frac{1}{2}+\frac{1}{2^{m+;1}}$ and $\gamma_1= \frac{1}{2}+\frac{1}{2^{m+1}}$.

\begin{lemma}
  \label{lemma:gip}
  Let $m$ and $\sigma_{m,b}$ be as in the above paragraph. Then we have
 \[ 
    \RII_\rho(\GIP_n)=\asOm{\min_b\col_{\rho\otimes\sigma_{m,b}}(n)},   \]
  where the constant factor in the asymptotic notation
  does not depend on~$n$ or~$\rho$.
\end{lemma}

	As we noted before, the reason $\col_{\rho\otimes\sigma_{m,b}}(n)$ as opposed to $\col_{\rho}(n)$ appear in the above is that starting from the second step of our induction we will fix the inner product of some of the $x_i\cdot y_i$ for some $i\in L_t$ which the players could presumably  take advantage of as an extra source of shared randomness. In the next section, we show that with our choice of parameter $m=8\log n$ $\min\{\col_{\rho\otimes\sigma_{m,0}}(n),\col_{\rho\otimes\sigma_{m,1}}(n)\}=\Omega(\col_{\rho}(n))$, and hence the players are in fact unable to take advantage of the extra shared randomness caused by conditioning; combined with Lemma \ref{lemma:gip} that will finish the proof. 

\begin{remark}
We shall note that even in the midst of the hybrid argument, there is still a minor difference between the original source $\rho$ and the auxiliary new source $\sigma_{m,b}$, in that unlike in the case of $\rho$, the players can only sample $\sigma_{m,b}$  for limited number of times, i.e.~$t \leq 2n/3$. However,  we are not able to, or rather do not need, to exploit this fine difference in the proof.	
\end{remark}

\begin{proof}Note that if $\min_{b}\col_{\rho\otimes \sigma_b}(n) \leq C\log n$ (for suitably large constant  $C$ to be determined later) then the result immediately follows from Proposition \ref{proposition:gip-lower-bound-log}. So in what follows we focus on the other case.

We denote the shared randomness between Alice and Bob  by $(R_A,R_B)\sim\rho^{\otimes \ell}$ and their input by $(x,y)\in(\F_2^m)^n\times(\F_2^m)^n$. We will denote Alice's and Bob's messages to the referee by $f(R_A,x), g(R_B, y)\in\{0,1\}^{CC_{\mathcal P}}$ respectively. For~$x,y\in(\F_2^m)^n$, $u\in U^\ell$, $v\in V^\ell$, and~$s,t\in\{0,1\}^{CC_{\mathcal P}}$,
let~$\pi_{x,y}(u,v,s,t)$ be the conditional probability
that the values of the shared random variables are~$(u,v)$
and the messages from Alice and Bob are~$s$ and~$t$, respectively,
given the input~$(x,y)$. Note that this probability is well-defined
even if the input~$(x,y)$ does not satisfy the promise of~$\GIP_n$.
Then~$\pi_{x,y}$ is a probability distribution
on~$U^\ell\times V^\ell\times\{0,1\}^{CC_{\mathcal P}} \times\{0,1\}^{CC_{\mathcal P}}$.

For~$L\subseteq[n]$ and~$b\in\{0,1\}$, let
\[
  S_{b,L} = \{(x,y)\in(\F_2^m)^n\times(\F_2^m)^n\colon
  x_i\cdot y_i=b\;(\forall i\in L)\}.
\]
Let~$T=\ceil{2n/3}$.
Then Lemma~\ref{lemma:RAIP} implies that
if~$b\in\{0,1\}$, $L\subseteq[n]$, and~$\abs{L}\le T$,
then there exists~$i\in[n]\setminus L$ such that
\[  \norm{
    (1-\gamma_b)\E_{(x,y)\in S_{b,L}\setminus S_{b,L\cup\{i\}}}{\pi_{x,y}}
    -
     \gamma_b \E_{(x,y)\in S_{b,L\cup \{i\} }}{\pi_{x,y}}
  }_1= O\left(\frac{CC_{\mathcal P}+ \log n}{n \col_{\rho\otimes \sigma_{b}}(n)}\right)+o\left(\frac{1}{n}\right),\]
  where here we took the advantage of the fact that $n-|L|=\Omega(n)$ to absorb the resulting constant, due to the shrinkage $n \mapsto n\setminus L$, in the asymptotic notation. Noting that $\mathbf{E}_{(x,y)\in S_{b,L}} =  (1-\gamma_b)\mathbf{E}_{(x,y)\in S_{b,L}\setminus S_{b,L\cup\{i\} } }  +\gamma_{b} \mathbf{E}_{(x,y)\in  S_{b, L\cup \{i\}}}$, we see that
  \begin{equation}
  	\norm{\E_{(x,y)\in S_{b,L}\setminus S_{b,L\cup\{i\}}}{\pi_{x,y}}   -
     2\gamma_b \E_{(x,y)\in S_{b,L\cup \{i\}}}{\pi_{x,y}}
  }_1= O\left(\frac{CC_{\mathcal P}+ \log n}{n \col_{\rho\otimes \sigma_{b}}(n)}\right)+o\left(\frac{1}{n}\right).
  \end{equation}
Applying induction and the triangle inequality we see that  there exist  size $T$ sets $L_0,L_1\subseteq[n]$ such that 
\[
  \norm{
    (2\gamma_0)^T\E_{(x,y)\in S_{0,L_0}}{\pi_{x,y}}
    -
    (2\gamma_1)^T\E_{(x,y)\in S_{1,L_1}}{\pi_{x,y}}
  }_1=O\left(\frac{CC_{\mathcal P}+ \log n}{\min_b \col_{\rho\otimes \sigma_{b}}(n)}\right)+o(1).
\]
Noting that with our choice of parameters $(2\gamma_{0})^{T}$ and $(2\gamma_1)^{T}$ are both $1+o(1)$ the above implies
\[
  \norm{
    \E_{(x,y)\in S_{0,L_0}}{\pi_{x,y}}
    -
  \E_{(x,y)\in S_{1,L_1}}{\pi_{x,y}}
  }_1=O\left(\frac{CC_{\mathcal P}+ \log n}{\min_b \col_{\rho\otimes \sigma_{b}}(n)}\right)+o(1).
\]

On the other hand, every input~$(x,y)\in S_{b,L_b}$ is a valid input in support of ~$\GIP_n$, and therefore the correctness of the protocol implies that $\norm{
    \E_{(x,y)\in S_{0,L_0}}{\pi_{x,y}}
    -
  \E_{(x,y)\in S_{1,L_1}}{\pi_{x,y}}
  }_1 \geq \frac{2}{3}$. Plugging this in the above and using the assumption that $\min_{b} \col_{\rho\otimes \sigma_{b}} (n)$ is larger than sufficiently large constant multiple of $\log n$ the result follows---which is something we can assume without loss of generality as seen from the discussion in the beginning of the proof.
\end{proof}

\subsection{Removal of additional shared randomness caused by conditioning}\label{sec:removal}
In this Section we prove the next Lemma which Lemma \ref{lemma:gip} immediately imply Theorem~\ref{t_GIP}.
\begin{lemma} \label{lemma:tensoring-ip}
Let~$\sigma_{m,b}$ be the uniform distribution over $\{(x,y)\in\{0,1\}^m\times\{0,1\}^m\colon x\cdot y=b\}$ and let $\sigma$ be the either of $\sigma_{m,0}$ or $\sigma_{m,1}$. If $n\leq 2^{m/2-2}$  then for any bipartite distribution~$\rho$  we have~$\col_{\rho\otimes\sigma}(n)=\Omega(\col_\rho(n))$,
  where the constant factor is independent of  $m, n$ or $\rho$.
\end{lemma}
We need two somewhat more technical statements to prove this Lemma. The first one is a basic estimate on the correlation complexity $\sigma$ which follows from elementary Fourier analysis. 

\begin{claim} \label{claim:ip-cor}
  Let~$m\ge2$ and~$b\in\{0,1\}$.
  Let~$\sigma$ be the uniform distribution over $\{(u,v)\in\{0,1\}^m\times\{0,1\}^m\colon u\cdot v=b\}$.
  Then it holds that~$
    \Cor(\sigma)\le1/2^{m/2-1}
  $.
\end{claim}

\begin{proof}
For~$f\colon\{0,1\}^m\to\RR$, the Fourier transform of~$f$ is defined by
\[
  \hat{f}(v)=\frac{1}{2^m}\sum_{u\in\{0,1\}^m}(-1)^{u\cdot v}f(u).
\]
Parseval's identity states that
\[
  \frac{1}{2^m}\sum_{u\in\{0,1\}^m}f(u)^2 = \sum_{v\in\{0,1\}^m}\hat{f}(v)^2.
\]

Let~$\mu$ be the marginal distribution of~$\sigma$ on either part.
Let~$f,g\colon\{0,1\}^m\to\RR$,
and assume that~$\E_{u\sim\mu}{f(u)}=\E_{v\sim\mu}{g(v)}=0$
and that~$\E_{u\sim\mu}{f(u)^2}=\E_{v\sim\mu}{g(v)^2}=1$.
Because~$\mu(u)\ge1/(2^m+1)$ for all~$u\in\{0,1\}^m$, it holds that
\[
  \sum_{u\in\{0,1\}^m}f(u)^2
  \le
  (2^m+1)\E_{u\sim\mu}{f(u)^2}
  =
  2^m+1,
\]
and similarly~$\sum_{v\in\{0,1\}^m}g(v)^2\le2^m+1$.

Simple calculation shows that
\begin{align*}
  &
  \E_{(u,v)\sim\sigma}{f(u)g(v)}
  \\
  &=
  \frac{1}{2^{2m-1}+(-1)^b\cdot2^{m-1}}
  \sum_{v\in\{0,1\}^m}g(v)\left(
    \sum_{u:u\cdot v=b}f(u)-\frac{(-1)^b\cdot2^m+1}{2}\E_{u\sim\mu}{f(u)}
  \right)
  \\
  &=
  \frac{1}{2^{2m}+(-1)^b\cdot2^m}
  \sum_{v\in\{0,1\}^m}(2^m \hat{f}(v)-f(0))g(v)
  \\
  &=
  \frac{1}{2^m+(-1)^b}
  \sum_{v\in\{0,1\}^m} \hat{f}(v)g(v)
  -
  \frac{f(0)}{2^{2m}+(-1)^b\cdot2^m}
  \sum_{v\in\{0,1\}^m}g(v)
  \\
  &=
  \frac{1}{2^m+(-1)^b}
  \sum_{v\in\{0,1\}^m} \hat{f}(v)g(v)
  +
  \frac{(-1)^b f(0)g(0)}{2^{2m}+(-1)^b\cdot2^m}.
\end{align*}
By the Cauchy--Schwarz inequality, the summation in the first term is bounded as
\begin{align*}
  \sum_{v\in\{0,1\}^m} \hat{f}(v)g(v)
  &\le
  \sqrt{\sum_{v\in\{0,1\}^m}\hat{f}(v)^2}\sqrt{\sum_{v\in\{0,1\}^m}g(v)^2}
  \\
  &\le
  \frac{1}{2^{m/2}}
  \sqrt{\sum_{u\in\{0,1\}^m}f(u)^2}\sqrt{\sum_{v\in\{0,1\}^m}g(v)^2}
  \\
  &\le
  \frac{2^m+1}{2^{m/2}},
\end{align*}
and therefore we have that
\begin{equation}
  \E_{(u,v)\sim\sigma}{f(u)g(v)}
  \le
  \frac{1}{2^{m/2}}
  \cdot
  \frac{2^m+1}{2^m+(-1)^b}
  +
  \frac{(-1)^b f(0)g(0)}{2^{2m}+(-1)^b\cdot2^m}.
  \label{eq:ip-cor-1}
\end{equation}

If~$b=0$, then~$f(0)^2\le\E_{u\sim\mu}{f(u)^2}/\mu(0)=(2^m+1)/2$,
and similarly~$g(0)^2\le(2^m+1)/2$.
Therefore, (\ref{eq:ip-cor-1}) implies that
\[
  \E_{(u,v)\sim\sigma}{f(u)g(v)}
  \le
  \frac{1}{2^{m/2}}
  +
  \frac{1}{2^{m+1}}
  <
  \frac{1}{2^{m/2-1}}.
\]

If~$b=1$, then the vector~$0\in\{0,1\}^m$
is not in the support of~$\mu$, and therefore
we may assume that~$f(0)=0$ without loss of generality.
Then~(\ref{eq:ip-cor-1}) implies that
\[
  \E_{(u,v)\sim\sigma}{f(u)g(v)}
  <
  \frac{1}{2^{m/2}}
  \cdot
  \frac{2^m+1}{2^m-1}
  <
  \frac{1}{2^{m/2-1}}.
\]
\end{proof}

\begin{claim} \label{claim:cor-to-col}
  Let~$\rho$ be a bipartite distribution on~$U\times V$,
  and~$\sigma$ be a bipartite distribution on~$X\times Y$.
  Then for~$p$ with~$\Cor(\sigma)<p<1$,
  it holds that~$\agr_{\rho\otimes\sigma}(p)\ge\agr_\rho(p-\Cor(\sigma))$.
\end{claim}

\begin{proof}
Let~$c>\agr_{\rho\otimes\sigma}(p)$.
Let~$(\ell,f,g)$ be an agreement protocol for~$\rho\otimes\sigma$
with success probability~$p$ and cost at most~$c$.
By Lemmas~\ref{lemma:correlation-tensor},
it holds that~$\Cor(\sigma^{\otimes \ell})=\Cor(\sigma)$.

For~$u\in U^\ell$,
let~$F(u)=\E_{x\sim\sigma_X^{\otimes \ell}}{f(u,x)}$.
Similarly, for~$v\in V^\ell$,
let~$G(v)=\E_{y\sim\sigma_Y^{\otimes \ell}}{g(v,y)}$.

Fix~$u\in U^\ell$ and~$v\in V^\ell$.
Because~$f$ and~$g$ take values in~$[0,1]$,
their variances for fixed~$u$ and~$v$ are at most~$1$.
By Lemma~\ref{lemma:correlation-general}, we have that
\[
  \E_{(x,y)\sim\sigma^{\otimes \ell}}{f(u,x)g(v,y)}
  \le
  F(u)G(v)
  +
  \Cor(\sigma).
\]
Then we have that
\[
  p=\E_{(u,v,x,y)\sim\rho^{\otimes \ell}\otimes\sigma^{\otimes \ell}}{f(u,x)g(v,y)}
  \le
  \E_{(u,v)\sim\rho^{\otimes \ell}}{F(u)G(v)}
  +
  \Cor(\sigma).
\]
This means that~$(\ell,F,G)$ is an agreement protocol for~$\rho$
with cost at most~$c$ and success probability at least~$p-\Cor(\sigma)$,
and therefore~$\agr_\rho(p-\Cor(\sigma))\le c$.
Because this holds for any~$c>\agr_{\rho\otimes\sigma}(p)$,
we conclude that~$\agr_\rho(p-\Cor(\sigma))\le\agr_{\rho\otimes\sigma}(p)$.
\end{proof}

\begin{proof}[Proof of Lemma~\ref{lemma:tensoring-ip}]
  Because~$n\le2^{m/2-2}$,
  Lemma~\ref{claim:ip-cor} implies that~$\Cor(\sigma)\le1/(2n)$.
  Then by Lemma~\ref{claim:cor-to-col}, it holds that
  \[
    \agr_{\rho\otimes\sigma}\left(\frac{1}{n}\right)
    \ge
    \agr_\rho\left(\frac{1}{n}-\Cor(\sigma)\right)
    \ge
    \agr_\rho\left(\frac{1}{2n}\right)
    \ge
    \frac{\agr_\rho(1/n)}{2}.
  \]
  The lemma follows from Lemma~\ref{lemma:col-agr}.
\end{proof}

\bibliography{SR}
\appendix

\section{Relation to hypercontractivity}\label{appendix:hyperc}

As we discussed, collision and agreement complexities are instances of measures of quality of correlation, and as such one natural way to analyse collision complexity is to relate it to other more widely used measures of quality of correlation. Indeed, in the case of maximum correlation, this was done in Subsection \ref{sec:removal} and that connection, i.e.
\begin{equation}\label{eq:drawback}
  \agr_\rho(p) \ge \sqrt{p-\Cor(\rho)},
\end{equation}
played an important role in the final stage of our argument. 

In this Section, our goal is to investigate the connection between collision and agreement complexities and another widely used measure of quality of correlation, i.e.~hypercontractivity.  As we shall show (and not surprisingly) the information offered by hypercontractivity properties of a source $\rho$, when available, is usually more powerful that what can be understood from the simple information provided by the maximum correlation.  More precisely, although an upper bound on the maximum correlation of~$\rho$ can be useful for lower bounding the collision complexity when the size of $\rho$ is not too small compared to the domain size $\frac{1}{p}$ (the agreement parameter), the estimates such as (\ref{eq:drawback}) stop being useful 
for very small~$p$, which is necessary for an asymptotic lower bound
on the collision complexity~$\col_\rho(n)$ as~$n\to\infty$.

The organisation of this appendix is as follows. We first start by discussing  some of the elementary (and well-known) properties of hypercontractivity in the first subsection; the material here is standard but we included them for the convenience of the reader. Next, we establish the connection between the collision complexity and  hypercontractivity which is the central result of the appendix. Finally in the last subsection we give a quick application of the preceding result by proving some simple (but non-trivial) upper and lower bounds for the collision complexity of $\rho_{\it disj}$.

\subsection{Preliminaries on hypercontractivity}

A generalisation of H\"{o}lder's inequality
provides a way to quantify the independence of a bipartite distribution.

For~$p\ge1$, a probability distribution~$\mu$ on~$X$,
and a function~$f\colon X\to\CC$,
the~$L_p$-norm of~$f$ with respect to~$\mu$,
denoted by~$\norm{f}_{L_p(\mu)}$, is
\[
  \norm{f}_{L_p(\mu)} = \E_{x\sim\mu}{\abs{f(x)}^p}^{1/p}.
\]
For~$p=\infty$, the~$L_\infty$-norm of~$f$ with respect to~$\mu$
is defined by
\[
  \norm{f}_{L_\infty(\mu)} = \max_{x:\mu(x)>0}\abs{f(x)}.
\]
If~$1\le p,p'\le\infty$ and~$1/p+1/p'=1$,
then H\"{o}lder's inequality states that
for any probability distribution~$\rho$ on~$U\times V$
and any functions~$f\colon U\to\CC$ and~$g\colon V\to\CC$, it holds that
\[
  \abs{\E_{(u,v)\sim\rho}{f(u)g(v)}} \le \norm{f}_{L_p(\rho_U)}\norm{g}_{L_p'(\rho_V)},
\]
where~$\rho_U$ and~$\rho_V$ are the marginal distributions of~$\rho$
on~$U$ and~$V$, respectively.
If~$\rho$ is a product distribution, then it holds that
\[
  \abs{\E_{(u,v)\sim\rho_U\otimes\rho_V}{f(u)g(v)}}
  \le
  \norm{f}_{L_1(\rho_U)}\norm{g}_{L_1(\rho_V)}.
\]
We consider a property between these two cases.

\begin{definition}[Generalized $(p,q)$ H\"{o}lder inequality]
  Let~$p,q\ge1$.
  A probability distribution~$\rho$ on~$U\times V$
  is said to satisfy the \emph{generalised $(p,q)$ H\"{o}lder inequality}
  if for any functions~$f\colon U\to\CC$ and~$g\colon V\to\CC$, it holds that
  \[
    \abs{\E_{(u,v)\sim\rho}{f(u)g(v)}}
    \le
    \norm{f}_{L_p(\rho_U)}\norm{g}_{L_q(\rho_V)},
  \]
  where~$\rho_U$ and~$\rho_V$ are the marginal distributions of~$\rho$
  on~$U$ and~$V$, respectively.
\end{definition}

It is often more convenient to prove this via hypercontractivity of a linear operator.
Note that a bipartite probability distribution~$\rho$ on~$U\times V$
defines a stochastic channel from~$V$ to~$U$
and therefore also defines a linear operator from~$\CC^U$ to~$\CC^V$:
\[
  (T_\rho f)(v) = \sum_{u\in U}\frac{\rho(u,v)}{\rho_V(v)}f(u).
\]
For~$1\le p\le q$, this operator~$T_\rho$
is said to be \emph{$q$-to-$p$ hypercontractive}
with respect to distribution~$\rho_U$
if for any~$f\colon U\to\CC$, it holds that
\[
  \norm{T_\rho f}_{L_q(\rho_V)} \le \norm{f}_{L_p(\rho_U)}.
\]
In this case, we say that the bipartite distribution~$\rho$ itself
is~$q$-to-$p$ hypercontractive.
Note that the order of sets~$U$ and~$V$ matters in this terminology.
The key relation between the generalised H\"{o}lder inequality
and the hypercontractivity which we will use is the following:

\begin{lemma} \label{lemma:hypercontractive-holder}
  Let~$p,q,q'\in[1,\infty]$ satisfy~$p\le q$ and~$1/q+1/q'=1$.
  Let~$\rho$ be a probability distribution on~$U\times V$.
  If~$\rho$ is~$q$-to-$p$ hypercontractive,
  then~$\rho$ satisfies the~$(p,q')$ generalised H\"{o}lder inequality.
\end{lemma}

\begin{proof}
  For~$f\colon U\to\CC$ and~$g\colon V\to\CC$, it holds that
  \begin{align*}
    \abs{\E_{(u,v)\sim\rho}{f(u)g(v)}}
    &=
    \abs{\E_{v\sim\rho_V}{(T_\rho f)(v)g(v)}}
    \\
    &\le
    \norm{T_\rho f}_{L_q(\rho_V)}\norm{g}_{L_{q'}(\rho_V)}
    \\
    &\le
    \norm{f}_{L_p(\rho_U)}\norm{g}_{L_{q'}(\rho_V)},
  \end{align*}
  where the first inequality follows from H\"{o}lder's inequality
  and the second inequality follows
  from the~$q$-to-$p$ hypercontractivity of~$\rho$.
\end{proof}


Like the maximum correlation, hypercontractivity tensorizes.

\begin{lemma} \label{lemma:hypercontractive-tensor}
  Let~$1\le p\le q\le\infty$.
  For~$i\in[n]$, let~$\rho_i$ be a probability distribution on~$U_i\times V_i$.
  If~$\rho_i$ is~$q$-to-$p$ hypercontractive for all~$i\in[n]$,
  then~$\rho_1\otimes\dots\otimes\rho_n$ is also~$q$-to-$p$ hypercontractive.
\end{lemma}

A proof of Lemma~\ref{lemma:hypercontractive-tensor} is standard
and follows from the fact that~$(q/p)$-norm on vectors
satisfies the triangle inequality (Minkowski's inequality).
The proof is essentially the same
as that of Lemma~5.3 of~\cite{Janson97}
and of Proposition~3.11 of~\cite{MosODoOle10AnnMath}.

\subsection{Hypercontractivity implies high collision complexity}

The following lemma states
that hypercontractivity of~$\rho$
implies an asymptotic lower bound on the collision complexity.

\begin{lemma} \label{lemma:col-hyp}
  Let~$\rho$ be a probability distribution on~$U\times V$.
  Let~$1\le p\le q$,
  and let~$q'$ and~$c$ be~$1/q+1/q'=1$ and~$1/c=1/p+1/q'$.
  If~$\rho$ is~$q$-to-$p$-hypercontractive,
  then for any~$z\in[0,1]$, it holds that
  \[
    \agr_\rho(z)\ge\frac{(p^{1/p}q^{\prime 1/q'}z)^c}{c},
  \]
  and in particular,
  \[
    \col_\rho(n)=\Omega(n^{1-c}).
  \]
\end{lemma}
As a sanity check note that in the above since $p\geq 1$ and $q\leq \infty$ we have $\frac{1}{c}=\frac{1}{p}+\frac{1}{q'}\leq 2$. Hence, in the Lemma $c\geq \frac{1}{2}$ which is in agreement with what we expect from the first part of Fact \ref{fact:col}.

\begin{proof}
  Let~$(\ell,f,g)$ be an agreement protocol for~$\rho$
  with success probability~$z$,
  and let
  \begin{align*}
    a&=\E{f(u)}=\norm{f}_{L_1(\rho_U)},
    \\
    b&=\E{g(v)}=\norm{g}_{L_1(\rho_V)}.
  \end{align*}
  We will prove that
  \[
    a+b\ge\frac{1}{c}(p^{1/p}q^{\prime 1/q'}z)^c.
  \]

  Lemma~\ref{lemma:hypercontractive-tensor} implies
  that~$\rho^{\otimes \ell}$ is~$q$-to-$p$ hypercontractive,
  and therefore Lemma~\ref{lemma:hypercontractive-holder} implies
  that~$\rho^{\otimes \ell}$ satisfies the~$(p,q')$ generalised H\"{o}lder inequality:
  \[
    z=\E_{(u,v)\sim\rho^{\otimes \ell}}{f(u)g(v)}
    \le
    \norm{f}_{L_p(\rho_U^{\otimes \ell})}\norm{g}_{L_{q'}(\rho_V^{\otimes \ell})}.
  \]
  Because~$\norm{f}_{L_\infty(\rho_U)}\le1$,
  it holds that
  \[
    \norm{f}_{L_p(\rho_U^{\otimes \ell})}^p
    \le
    \norm{f}_{L_\infty(\rho_U^{\otimes \ell})}^{p-1}\norm{f}_{L_1(\rho_U^{\otimes \ell})}
    \le
    1\cdot a,
  \]
  and therefore~$a^{1/p}\ge\norm{f}_{L_p(\rho_U^{\otimes \ell})}$.
  Similarly, it holds that~$b^{1/q'}\ge\norm{g}_{L_{q'}(\rho_V^{\otimes \ell})}$,
  and therefore~$a^{1/p}b^{1/q'}\ge z$.

  By concavity of logarithm, it holds that
  \begin{align*}
    \log\bigl(c(a+b)\bigr)
    &=
    \log\frac{(1/p)\cdot pa+(1/q')\cdot q'b}{1/p+1/q'}
    \\
    &\ge
    \frac{(1/p)\log(pa)+(1/q')\log(q'b)}{1/p+1/q'}
    \\
    &=
    c\cdot\left((1/p)\log(pa)+(1/q')\log(q'b)\right)
  \end{align*}
  By taking the exponentials of both sides, we obtain that
  \[
    c(a+b)
    \ge
    (p^{1/p}q^{\prime 1/q'}a^{1/p}b^{1/q'})^c
    \ge
    (p^{1/p}q^{\prime 1/q'}z)^c.
  \]
  This means that the agreement complexity of~$\rho$
  satisfies that~$\agr_\rho(z)\ge(p^{1/p}q^{\prime 1/p'}z)^c/c$.
  By Lemma~\ref{lemma:col-agr},
  we have that~$\col_\rho(n)=\Omega(n\agr_\rho(1/n))=\Omega(n^{1-c})$.
\end{proof}


\subsection{Analysing the  collision complexity of $\rho_{disj}$}
\label{appendix:example}

In this section, we analyse the collision complexity of $\rho_{disj}$ and show that it is strictly in between the extremes of a constant and~$\sqrt{n}$. The lower bound gives us an opportunity to show the applicability of the techniques developed in the previous parts of the Appendix in this  simple setting.

\begin{proposition}
  \label{proposition:collision-example}
  Let~$\rho=\rho_{\it disj}$ be as in Example \ref{example:disj}. We have,
 \[
    \col_{\rho_{\it disj}}(n)= O(n^{\log_6 2})\cap\Omega(n^{1/4})
    \subseteq O(n^{0.387})\cap\Omega(n^{0.25}).
  \]
\end{proposition}

\subsubsection{Upper bound}

Let~$0<p<1$, and consider the following agreement protocol~$(\ell,f,g)$.
Let~$\ell=\floor{\log_6(1/p)}$.
Define
\[
  f(x_1,\dots,x_\ell)=\begin{cases}
    1, & x_1=\dots=x_\ell=1, \\
    0, & \text{otherwise.}
  \end{cases},
  \qquad
  g(y_1,\dots,y_\ell)=\begin{cases}
    1/2^\ell, & y_1=\dots=y_\ell=0, \\
    0, & \text{otherwise.}
  \end{cases}
\]
We claim that this agreement protocol has success probability at least~$p$
and cost less than~$6p^{\log_6 3}$.

Note that by definition of distribution~$\rho$,
we have that
\[
  \Pr[x_1=\dots=x_\ell=1 \land y_1=\dots=y_\ell=0]=\frac{1}{3^\ell}.
\]
Therefore, it holds that the success probability of the protocol~$(\ell,f,g)$ is
\begin{align*}
  \E{f(x_1,\dots,x_\ell)g(y_1,\dots,y_\ell)}
  =
  \frac{1}{2^\ell}\Pr[x_1=\dots=x_\ell=1 \land y_1=\dots=y_\ell=0]
  =
  \frac{1}{2^\ell}\cdot\frac{1}{3^\ell}
  \ge
  p.
\end{align*}

The cost of this protocol is equal to
\[
  \left(\frac13\right)^\ell+\frac{1}{2^\ell}\cdot\left(\frac23\right)^\ell
  =
  \frac{2}{3^\ell}
  =
  \frac{6}{3^{\ell+1}}
  <
  \frac{6}{3^{\log_6(1/p)}}
  =
  6p^{\log_6 3}.
\]

The upper bound~$\col_{\rho_{\it disj}}(n)= O(n^{\log_6 2})$
follows from Lemma~\ref{lemma:col-agr}.

\subsubsection{Lower bound}

We claim that~$\rho=\rho_{\it disj}$ is~$3$-to-$3/2$ hypercontractive.
Let~$f\colon\{0,1\}\to\CC$,
and we will prove that~$\norm{T_\rho f}_3\le\norm{f}_{3/2}$.
Let~$\alpha=\abs{f(0)}$ and~$\beta=\abs{f(1)}$.
Because
\begin{align*}
  (T_\rho f)(0) &= \frac{f(0)+f(1)}{2}, \\
  (T_\rho f)(1) &= f(0),
\end{align*}
we have that
\begin{align*}
  \norm{T_\rho f}_3
  &=
  \left(\frac23\abs*{\frac{f(0)+f(1)}{2}}^3+\frac13\abs{f(0)}^3\right)^{1/3}
  \le
  \left(\frac23\left(\frac{\alpha+\beta}{2}\right)^3+\frac13\alpha^3\right)^{1/3},
  \\
  \norm{f}_{3/2}
  &=
  \left(\frac23\alpha^{3/2}+\frac13\beta^{3/2}\right)^{2/3}.
\end{align*}
Simple calculations show that
\[
  \norm{f}_{3/2}^3-\norm{T_\rho f}_3^3
  =
  \frac{(\sqrt{\alpha}-\sqrt{\beta})^4(\alpha+4\sqrt{\alpha\beta}+\beta)}{36}
  \ge0,
\]
establishing the claim that~$\rho=\rho_{\it disj}$ is~$3$-to-$3/2$ hypercontractive.
The lower bound~$\col_{\rho_{\it disj}}(n)=\Omega(n^{1/4})$
follows from Lemma~\ref{lemma:col-hyp}.

\end{document}